\documentclass[acmsmall,10pt,screen]{acmart}
\acmJournal{TOPLAS}
\acmYear{2021}
\acmISBN{} 
\acmDOI{} 
\startPage{1}

\setcopyright{none}

\usepackage{booktabs}   
\usepackage{subcaption} 

\usepackage{amsmath}
\usepackage{amssymb}
\usepackage{bm}

\usepackage{mathtools}
\DeclarePairedDelimiterX\set[1]\lbrace\rbrace{\def\given{\;\delimsize\vert\;}#1}

\usepackage[nocomma]{optidef}
\usepackage{algorithm}
\usepackage{enumitem}
\usepackage{algorithmicx}
\usepackage{algpseudocode}
\usepackage{pifont}
\usepackage{listings}
\usepackage[symbol]{footmisc}
\usepackage{xcolor}
\newcommand{\greencheck}{{\color{green!60!gray}\checkmark}}
\newcommand{\myred}[1]{{\color{red!60!gray} #1}}
\usepackage[frozencache=true, cachedir = \detokenize{./minted-cache}]{minted} 
\renewcommand\listingscaption{Code}

\usepackage{multirow}

\usepackage{enumitem}
\newlist{steps}{enumerate}{1}
\setlist[steps, 1]{label = \textbf{Step} \arabic*:}
\newcommand{\oomit}[1]{}

\usepackage[algo2e, ruled, lined, boxed, linesnumbered, commentsnumbered]{algorithm2e}
\SetKwComment{algocomment}{$\triangleright$\ }{}

\SetCommentSty{mycommfont}

\renewcommand{\thefootnote}{\fnsymbol{footnote}}

\newtheorem{lemma}{Lemma}

\newtheorem{theorem}{Theorem}
\newtheorem{remark}{Remark}
\newtheorem{definition}{Definition}
\newtheorem{example}{Example}

\usepackage[capitalize]{cleveref}
\crefname{lemma}{Lem.}{Lem.}
\creflabelformat{lemma}{#2\textup{#1}#3}
\crefname{example}{Exmp.}{Exmp.}
\creflabelformat{example}{#2\textup{#1}#3}
\crefname{section}{Sect.}{Sect.}
\creflabelformat{section}{#2\textup{#1}#3}
\crefname{appendix}{Appx.}{Appx.}
\creflabelformat{appendix}{#2\textup{#1}#3}
\crefname{definition}{Def.}{Def.}
\creflabelformat{definition}{#2\textup{#1}#3}
\crefname{theorem}{Thm.}{Thm.}
\creflabelformat{theorem}{#2\textup{#1}#3}
\crefname{proposition}{Prop.}{Prop.}
\creflabelformat{proposition}{#2\textup{#1}#3}
\crefname{corollary}{Cor.}{Cor.}
\creflabelformat{corollary}{#2\textup{#1}#3}
\crefname{problem}{Problem}{Problem}
\creflabelformat{problem}{#2\textup{#1}#3}
\crefname{algorithm}{Alg.}{Alg.}
\creflabelformat{algorithm}{#2\textup{#1}#3}
\crefname{enumi}{}{}

\newcommand{\revise}[1]{#1}

\newtheorem{manualexampleinner}{Example}
\newenvironment{manualexample}[1]{%
  \renewcommand\themanualexampleinner{#1}%
  \manualexampleinner
}{\endmanualexampleinner}

\newcommand{\qm}{\textbf{Q}}
\newcommand{\mm}{\textbf{m}}

\newcommand{\Real}{\mathbb{R}}
\newcommand{\Nat}{\mathbb{N}}

\newcommand{\tool}[1]{\textsc{#1}}

\newcommand{\xmark}{\ding{55}}%

\begin{document}

\title[Synthesizing Invariants for Polynomial Programs by Semidefinite Programming]{Synthesizing  Invariants for Polynomial Programs by Semidefinite Programming}     


\author{Hao Wu}
\orcid{0000-0001-9368-4744}
\affiliation{
    \institution{SKLCS, Institute of Software, University of Chinese Academy of Sciences}
    \city{Beijing}
    \country{China}
}
\email{wuhao@ios.ac.cn}

\author{Qiuye Wang}
\orcid{0000-0001-5138-3273}
\affiliation{
    \institution{SKLCS, Institute of Software, University of Chinese Academy of Sciences}
    \city{Beijing}
    \country{China}
}
\email{wangqye@ios.ac.cn}

\author{Bai Xue}
\orcid{0000-0001-9717-846X}
\affiliation{
    \institution{SKLCS, Institute of Software, Chinese Academy of Sciences \& University of CAS}
    \city{Beijing}
    \country{China}
}
\email{xuebai@ios.ac.cn}

\author{Naijun Zhan}
\orcid{0000-0003-3298-3817}
\affiliation{
\department{School of Computer Science}
    \institution{Peking University \&}
    \institution{SKLCS, Institute of Software, Chinese Academy of Sciences}
    \city{Beijing}
    \country{China}
}
\email{znj@ios.ac.cn}

\author{Lihong Zhi}
\affiliation{
    \institution{KLMM, Academy of Mathematics and Systems Science, Chinese Academy of Sciences \& University of CAS}
    \city{Beijing}
    \country{China}
}
\email{lzhi@mmrc.iss.ac.cn}

\author{Zhi-Hong Yang}
\orcid{0000-0002-7489-9237}
\affiliation{
    \institution{School of Mathematics and Statistics, Central South University}
    \state{Changsha}
    \country{China}
}
\email{yangzhihong@csu.edu.cn}

\renewcommand{\shortauthors}{Wu et al.}

\begin{abstract}
Constraint-solving-based program invariant synthesis takes a parametric invariant template and encodes the (inductive) invariant conditions into constraints.
The problem of characterizing the set of all valid parameter assignments is referred to as the \emph{strong invariant synthesis problem}, while the problem of finding a concrete valid parameter assignment is called the \emph{weak invariant synthesis problem}.
\revise{
For both problems, the challenge lies in solving or reducing the encoded constraints, which are generally non-convex and lack efficient solvers. 
Consequently, existing works either rely on heuristic optimization techniques (such as bilinear matrix inequalities) or resort to general-purpose solvers (such as quantifier elimination), leading to a trade-off between completeness and efficiency.
}

\revise{
In this paper, 
we propose two novel algorithms for synthesizing invariants of polynomial programs using semidefinite programming (SDP):
(1) The Cluster algorithm targets the strong invariant synthesis problem for polynomial invariant templates.
Leveraging robust optimization techniques, it solves a series of SDP relaxations and yields a sequence of increasingly precise under-approximations of the set of valid parameter assignments.
We prove the algorithm's soundness, convergence, and weak completeness under a specific robustness assumption on templates. 
Moreover, the outputs can simplify the weak invariant synthesis problem.
(2) The Mask algorithm addresses the weak invariant synthesis problem in scenarios where the aforementioned robustness assumption does not hold, rendering the Cluster algorithm ineffective.
It identifies a specific subclass of invariant templates, termed masked templates, involving parameterized polynomial equalities and known inequalities. 
By applying variable substitution, the algorithm transforms constraints into an equivalent form amenable to SDP relaxations.
}
Both algorithms have been implemented and demonstrated superior performance compared to state-of-the-art methods in our empirical evaluation.
\end{abstract}

\ccsdesc[500]{Theory of computation~Invariants}
\ccsdesc[300]{Theory of computation~Program verification}
\ccsdesc[100]{Theory of computation~Logic and verification}
\ccsdesc[300]{Mathematics of computing~Semidefinite programming}

\keywords{program verification, invariant synthesis, sum-of-squares relaxations, semidefinite programming}

\maketitle

\section{Introduction}
\label{sec:introduction}

The dominant approach to program verification is \emph{Floyd-Hoare-Naur's inductive assertion method}~\cite{floyd67,hoare69,naur66}, which is based on Hoare logic~\cite{hoare69}.
The central concept of Hoare logic is the \emph{Hoare triple}
with the form 
\begin{equation*}
\{P\}C\{Q\},    
\end{equation*}
where $C$ is a piece  of program to be verified, $P$ is the precondition, and $Q$ is the postcondition.
The program $C$ is said to be \emph{partially correct} with respect to specifications $P$ and $Q$ if, assuming the precondition $P$ holds before executing $C$ and the program $C$ terminates, then the postcondition $Q$ will hold upon the completion of $C$.



In Hoare logic, an \emph{invariant} is an assertion associated with a particular program location, and it holds  whenever the location is reached during program execution.
An \emph{inductive invariant} is a specific type of invariant that satisfies the inductive property: 
If the assertion holds at a program location, 
then it is preserved during subsequent visits to that location.
The difference between invariants and inductive invariants is discussed in detail in~\cite{Manna04}. 
For the purpose of this paper, we will solely focus on inductive invariants, and for simplicity, we will refer to them simply as "invariants", unless otherwise stated.

Invariant generation stands as a crucial aspect of Hoare-style program verification.
The effectiveness of the verification process heavily relies on the ability to discover appropriate invariants that accurately capture the behavior and properties of the program throughout its execution.
Though this has been shown to be undecidable in general~\cite{OS04ICALP},
many efforts have been put into this area, resulting in various invariant synthesis techniques, 
including approaches based on 
constraint solving (discussed below),
recurrence analysis~\cite{kincaid2017nonlinear, humenberger2018invariant}, 
abstract interpretation~\cite{rck04a, rodriguez2007automatic}, 
Craig interpolation~\cite{LS17, Gan2020},
machine learning~\cite{si2020code2inv, yao2020learning}, 
and so on. 

Constraint-solving-based invariant synthesis, also known as template-based invariant synthesis, is promising and therefore prominent for discovering program invariants.
The general workflow of these methods can be summarized as follows:
The algorithm takes a user-specified parametric formula as input, encodes the invariant conditions into constraints, and attempts to reduce or solve these constraints.
The \emph{strong invariant synthesis problem} aims to characterize the set of all valid parameter assignments that satisfy the invariant conditions. This typically involves reducing the initial constraints into new constraints solely on the parameters.
In contrast, the \emph{weak invariant synthesis problem} seeks to find a single concrete invariant that satisfies the constraints. This usually involves solving the constraints to identify a valid parameter assignment that defines a specific invariant.

In this paper, we consider both the strong and the weak invariant synthesis problem for polynomial programs over real-valued variables, where the invariant templates are given as conjunctions of polynomial inequalities.
In this setting, the conditions for a template to be an invariant can be expressed as a first-order logic formula. 
It is worth noting that if both the program and the specifications (precondition and postcondition) are polynomial, the truth of the formula is decidable, as per Tarski's theorem~\cite{Tarski51}.
For the weak invariant synthesis problem, one major approach is to use Putinar's Positivstellensatz (see~\cref{thm:putinar}) to transform the invariant conditions into constraints containing sum-of-squares (SOS) polynomials, i.e., bilinear matrix inequalities (BMI).
However, solving the resulting constraints is still NP-hard~\cite{toker95bmi, blondel95}, and hence existing works mostly rely on general-purpose solvers for non-linear real arithmetic~\cite{YZZ+10FCSC, chatterjee2020polynomial, GHM+23OOPSLA} or heuristic strengthening strategies~\cite{LWY+14FCS, AGM15SAS, Cousot05VMCAI} to tackle them.
For the strong invariant synthesis problem, only a general doubly exponential upper bound based on quantifier elimination is known~\cite{kapur05inv}.

\revise{
\paragraph{Contributions} 
We propose two novel SDP-based approaches for the invariant synthesis problem of polynomial programs over real-valued variables.

The Cluster algorithm (\cref{sec:under-approximation} and \cref{sec:extension}) tackles the strong invariant synthesis problem for invariant templates expressed as conjunctions of parameterized polynomial inequalities.
To characterize the set of valid assignments of parameters $\bm{a}$, called the \emph{valid set}, we employ Lasserre's technique~\cite{lasserre15} to construct a series of SOS relaxations of the invariant conditions.
Notably, our encoding of SOS relaxations can be solved as SDPs, which is not the case as in~\cite{GHM+23OOPSLA, chatterjee2020polynomial}.
For any $D\in \mathbb{N}$,
our Cluster algorithm produces a series of under-approximations $R_{I,1}, R_{I,2},\dots, R_{I,D}$ of the valid set.
For each $d$ such that $1\le d\le D$, the $d$th under-approximation is defined by $R_{I,d}=\{\bm{a}\mid h_d(\bm{a})\le 0\}$, where $h_d(\bm{a})$ is a polynomial of degree at most $d$. 
As $D$ goes to infinity, the sequence of under-approximations converges to the valid set.
Furthermore, under the robustness assumption that the valid set has an interior point, we establish a semi-completeness result: 
for sufficiently large $D$, the Cluster algorithm will produce a non-empty under-approximation $R_{I,d}$ of the valid set for some $1\le d\le D$.
In such cases,
the weak invariant synthesis problem reduces to solving $h_d(\bm{a})\le 0$.


The Mask algorithm (\cref{sec:masked}) is designed for the weak invariant synthesis problem in scenarios where the Cluster algorithm fails due to the violation of the robustness assumption.
Such scenarios often arise when the invariant templates include equalities.
The Mask algorithm focuses on \emph{masked templates},
a specific subclass of invariant templates containing parameterized polynomial equalities and known inequalities.
By using variable substitution, the Mask algorithm transforms the invariant conditions into constraints that again allow for a hierarchy of SDP relaxations.

The two algorithms have been implemented and tested on two sets of benchmarks, depending on whether the invariant templates include equalities.
Compared with state-of-the-art constraint-solving-based and learning-based methods, both of our approaches demonstrate advantages in terms of effectiveness and efficiency.
}

\paragraph{Structure}
The rest of this paper is organized as follows: 
In \cref{sec:pre}, 
we introduce basic notions and algebraic tools that will be used. 
\cref{sec:under-approximation} explains Lasserre's technique and proposes the Cluster algorithm.
\cref{sec:extension} discusses additional extensions to enhance the expressiveness of the algorithm.
\cref{sec:masked} introduces the definition of masked templates and presents the Mask algorithm.
We report the experimental results in \cref{sec:experimentsnew} and discuss related work in \cref{sec:related}. 
Finally, \cref{sec:conclusion} concludes the paper.
\section{Preliminaries}
\label{sec:pre}

In the following, we first fix the notation used throughout the rest of this paper.
In \cref{sec:pre-sos}, we give a brief introduction to SOS relaxations, which serves as the fundamental technique in our algorithms.
In \cref{sec:2-2}, we formally define the invariant synthesis problems of interest.

\paragraph{Basic Notations}
The following basic notions will be used.   
Let $\mathbb{R}$ and $\mathbb{N}$ denote the set of real numbers and the set of natural numbers, respectively.
We use boldface letters to denote vectors (such as $\bm{x}$), 
vector-valued functions (such as $\bm{f}(\bm{x})$), and vector of constants (such as $\bm{0}$).
The comparison between vectors is element-wise.
\revise{
Given a function $f:\mathbb{R}^n\to \mathbb{R}$, the $0$-sublevel set of $f$ is the set $\{\bm{x}\in \mathbb{R}^n\mid f(\bm{x})\le 0\}$.
For a vector $\bm{x}=(x_1,\dots,x_n)\in \Real^n$,
we use $\|\bm{x}\|_1= \sum_{i=1}^n |x_i|$ and $\|\bm{x}\|_2= \sqrt{\sum_{i=1}^n x_i^2}$ to denote its $l_1$-norm and $l_2$-norm, respectively.
}
We say $\bm{x}_0\in \Real^n$ is an interior point of a set $S\subseteq \Real^n$ if there exists $\epsilon>0$ such that $\bm{x}\in S$ for all $\bm{x}$ satisfying $\|\bm{x}-\bm{x}_0\|_2\le \epsilon$.
We use $\mu$ to denote the Lebesgue measure.
Given a hyper-rectangle $C=[a_1,b_1]\times \dots \times [a_n,b_n]\in \mathbb R^n$ with $-\infty<a_i<b_i<\infty$, the volume of $C$ is $\mu(C)=\prod_{i=1}^n (b_i-a_i)$.

We will also use the following notations from real algebraic geometry.
Let $\bm{x}=(x_1,\dots,x_n)$ denote a vector of variables in $\Real^n$.
$\mathbb{R}[\bm{x}]$ denotes the ring of polynomials in variables $\bm{x}$, and $\mathbb{R}^d[\bm{x}]$ denotes the set of polynomials of degree less than or equal to $d$ in variables $\bm{x}$, where $d \in \mathbb{N}$. 
For convenience, we do not explicitly distinguish a polynomial $p \in \mathbb{R}[\bm{x}]$ and the function $p(\bm{x})$ it introduces.
A basic semialgebraic set $\mathcal{K}\subseteq \mathbb R^n$ is of the form $\{\bm{x}\in \mathbb{R}^n \mid p_1(\bm{x}) \diamond 0, \dots, p_m(\bm{x}) \diamond 0\}$, 
where $p_i(\bm{x}) \in \mathbb{R}[\bm{x}]$ and 
each $\diamond$ can be one of $\{ < , \le, = , \ge, > \}$.  
A basic closed semialgebraic set is of the form $\{\bm{x}\in \mathbb{R}^n \mid p_1(\bm{x}) \ge 0, \dots, p_m(\bm{x}) \ge 0\}$.
A semialgebraic set is of the form $\bigcup_{i=1}^{n} \mathcal{K}_i$,  where each $\mathcal{K}_i$ is a basic semialgebraic set. 
We say a polynomial $p(\bm{x})$ is non-negative (resp. strictly positive) over $\mathcal{K}$ if $p(\bm{x})\ge 0$ (resp. $p(\bm{x})>0$) for all $\bm{x}\in \mathcal{K}$.

\subsection{SOS Relaxations}
\label{sec:pre-sos}

The SOS relaxation is a well-established technique in polynomial optimization.
Its basic idea is to approximate a non-convex polynomial optimization problem by a sequence of convex optimization problems.
In this part, we introduce necessary concepts related to this technique and demonstrate a typical application.
For more comprehensive technical details, please  refer to~\cite{marshall2008book,lasserre2009book}.

\paragraph{Putinar's Representation Theorem}
A polynomial $p(\bm{x})\in \mathbb R[\bm{x}]$ is said to be an SOS polynomial if it can be expressed as $p(\bm{x})=\sum_{i=1}^m p_i(\bm{x})^2$, where $p_i(\bm{x})\in \mathbb R[\bm{x}]$ and $m\in \mathbb N$.
Similar to $\mathbb R[\bm{x}]$ and $\mathbb R^d[\bm{x}]$, we use $\Sigma[\bm{x}]$ and $\Sigma^{d}[\bm{x}]$ to denote the set of SOS polynomials and the set of SOS polynomials of degree less than or equal to $d$ in variables $\bm{x}$, respectively. 
Since the degree of an SOS polynomial must be even, we have $\Sigma^{2d}[\bm{x}]=\Sigma^{2d+1}[\bm{x}]$ for any $d\in \Nat$.

\begin{definition}[Quadratic Module~\cite{marshall2008book}]
A subset $\qm$ of $\mathbb{R}[\bm{x}]$ is called a \emph{quadratic module} if it contains $1$ and is closed under addition and multiplication with squares, i.e.,
\begin{equation*}
    1 \in \qm, \quad \qm + \qm \subseteq \qm, \quad \text{and} \quad p^2\qm \subseteq \qm \text{ for all } p\in \mathbb{R}[\bm{x}].
\end{equation*}
\end{definition}

\revise{
Given a vector of polynomials 
\begin{equation}\label{eq:p}
    \bm{p}(\bm{x})=\big(p_1(\bm{x}),  \dots,p_m(\bm{x})\big),
\end{equation}
and let $\mathcal{K}$ be the basic closed semialgebraic set defined by $\bm{p}(\bm{x})\ge \bm{0}$, i.e.,
\begin{equation}\label{eq:K}
\mathcal K=\set[\big]{\bm{x}\in \Real^n \given p_1(\bm{x})\ge 0, \dots, p_m(\bm{x})\ge 0}    
\end{equation}
we define the quadratic module generated by $\bm{p}$ as follows:
}

\begin{definition} \label{def:qm}
Let $\bm{p}(\bm{x})$ be defined as above, we denote by $\qm(\bm{p})$ the smallest quadratic module generated by polynomials in $\bm{p}$, i.e.,
\begin{equation*}
    \qm(\bm{p}) = 
    \set[\bigg]{ \sigma_0 + \sum_{i=1}^{m} \sigma_i p_i \given \sigma_i \in \Sigma[\bm{x}] }.
\end{equation*}
Furthermore, a quadratic module $\qm(\bm{p})$ is called \emph{Archimedean}, or satisfies the \emph{Archimedean condition}, 
if $N - \|\bm{x}\|_2^2 \in \qm(\bm{p})$ for some constant $N \in \mathbb{N}$.
\end{definition}

\revise{
Since SOS polynomials in $\Sigma[\bm{x}]$ are non-negative over $\Real^n$,
it is easy to see that the following lemma holds.

\begin{lemma}\label{lem:qm}
    Given $\bm{p}(\bm{x})$ and $\mathcal K$ as defined in \cref{eq:p} and \cref{eq:K}, respectively.
    If polynomial $f\in \qm(\bm{p})$, then $f(\bm{x})$ is \emph{non-negative} over $\mathcal{K}$.
\end{lemma}
}

By \cref{lem:qm}, the Archimedean condition $N-\|\bm{x}\|_2^2=N-\sum_{i=1}^n x_i^2 \ge 0$ over $\mathcal{K}$ implies that $\mathcal{K}$ is bounded.
On the other hand, if $\mathcal{K}$ is known to be bounded in a $n$-dimensional ball $\{\bm{x}\in \Real^n \mid N - \|\bm{x}\|_2^2 \ge 0 \}$ for some $N\in \Nat$, we can introduce a redundant polynomial $p_{m+1}=N - \|\bm{x}\|_2^2$ and denote $\bm{p}'(\bm{x})=\big(p_1(\bm{x}), \dots,p_m(\bm{x}), p_{m+1}(\bm{x})\big)$, then $\qm(\bm{p}')$ is Archimedean.
Therefore, the Archimedean condition can be intuitively interpreted as compactness, i.e., being closed and bounded.

Now, we present an important representation theorem in real algebraic geometry, called Putinar's Positivstellensatz~\cite{putinar1993},
which provides a characterization of polynomials that are \emph{locally positive} over a compact basic closed semialgebraic set.
For convenience, we use the formulation in~\cite{lasserre2009book}.

\begin{theorem}[Putinar's Positivstellensatz~{\cite[Thm.~2.14]{lasserre2009book}}]
\label{thm:putinar}
    Given $\bm{p}(\bm{x})$ and $\mathcal K$ as defined in \cref{eq:p} and \cref{eq:K}, respectively.
    If $\qm(\bm{p})$ is Archimedean and a polynomial $f\in \mathbb R[\bm{x}]$ is \emph{strictly positive} over $\mathcal{K}$,
    then $f\in \qm(\bm{p})$.
\end{theorem}

\begin{remark}\label{remark:homo} 
When the Archimedean condition is violated (e.g., $\mathcal K$ is unbounded), we need to resort to 
variants of~\cref{thm:putinar}, such as the homogenization formulation~\cite{huang23mp} and Putinar–Vasilescu Positivstellensatz~\cite{mai22mp}. 
\revise{
For example, the homogenization formulation has been applied in invariant generation for continuous-time dynamical systems~\cite{wu24fm-bc} and Craig interpolation synthesis~\cite{wu24fm-craig}, a logic inference technique  very related to invariant generation.
}
While these extensions allow handling unbounded domains, they come at the cost of increased complexity in the resulting constraints.
Nevertheless, these constraints can still be efficiently solved (as SDP, introduced later). 
To maintain focus on the core technique presented in this paper, we will limit our study to the scenario where program variables are bounded. 
Moreover, our algorithms remain sound even when this condition is violated.
\end{remark}

\revise{
\paragraph{SOS Relaxations}
We demonstrate the use of \cref{thm:putinar} on a typical template-based synthesis problem.
Let $l(\bm{a},\bm{x})$ be a parameterized polynomial that is linear in parameters $\bm{a}$.
For example, one can consider $l(\bm{a},\bm{x})$ to be a polynomial template of degree $d$ in variable $\bm{x}$ and the parameters $\bm{a}$ represent the unknown coefficients, i.e., 
\begin{equation}\label{eq:template}
    l(\bm{a},\bm{x})=\sum_{\|\bm{\beta}\|_1\le d} \bm{a}_{\bm{\beta}} \bm{x}^{\bm{\beta}}.
\end{equation}
where $\bm{\beta} = (\beta_1,\dots,\beta_{n})\in \Nat^{n}$ are exponents such that $\|\bm{\beta}\|_1=\sum_{j=1}^{n} \beta_j \le d$.
For example, when $\bm{x}=(x_1,x_2)$, a template could be $l(\bm{a},\bm{x})=a_1 x_1^2+a_2 x_1x_2+a_3 x_2^2+a_4x_1+a_5x_2+a_6$, which represents any polynomial in $(x_1,x_2)$ of degree not exceeding $2$.

Suppose we want to find a valid parameter assignment $\bm{a}_0$ such that $l(\bm{a}_0,\bm{x})\le 0$ over the compact basic closed semialgebraic set $\mathcal K$ as defined in \cref{eq:K},
the problem can be formalized as follows: 
\begin{equation}\label{eq:pop}
\begin{aligned}
    \text{find} \quad& \bm{a} \\
    s.t. \quad& \forall \bm{x}. \bigwedge_{i=1}^m p_i(\bm{x})\ge 0\implies l(\bm{a},\bm{x}) \le 0.
\end{aligned}
\end{equation}

Since $\mathcal K$ is bounded, we can assume that the quadratic module $\qm(\bm{p})$ is Archimedean.
By applying \cref{thm:putinar},
one can construct the following program:
\begin{equation}\label{eq:sos}
\begin{aligned}
    \inf \quad& \gamma \\
    \textit{s.t.} \quad& \gamma - l(\bm{a},\bm{x}) = \sigma_0(\bm{x}) +\sum_{i=1}^m \sigma_i(\bm{x}) \cdot p_i(\bm{x}),\\
    &\quad  
    \sigma_0 \in \Sigma[\bm{x}], 
    \sigma_i \in \Sigma[\bm{x}], \quad \text{for } i= 1,\dots,m,
\end{aligned}
\end{equation}
where $\gamma$ is a newly introduced variable serving as an upper bound of $l(\bm{a},\bm{x})$ over $\mathcal{K}$.
Let $\gamma^*$ denote the optimal value of Prog.~(\ref{eq:sos}).
The relation between Prog.~(\ref{eq:pop}) and Prog.~(\ref{eq:sos}) is stated by the following two theorems.

\begin{theorem}[Soundness]\label{thm:pre-sound}
If $\gamma^*< 0$ or the infimum $\gamma^*=0$ is attainable in Prog.~(\ref{eq:sos}), then Prog.~(\ref{eq:pop}) is feasible.
\end{theorem}

\begin{proof}
Using the assumption, there must exist a feasible solution $\gamma_0\le 0$ and $\bm{a}_0$ such that the constraint in Prog.~(\ref{eq:sos}) holds. 
By \cref{lem:qm}, we have $\gamma_0-l(\bm{a}_0,\bm{x})\ge 0$ when $\bm{x}\in\mathcal{K}$.
Therefore, $\bm{a}_0$ is a solution to Prog.~(\ref{eq:pop}).
\end{proof}

\begin{theorem}[Semi-completeness]\label{thm:pre-complete}
If Prog.~(\ref{eq:pop}) is feasible,
then $\gamma^*\le 0$ in Prog.~(\ref{eq:sos}).
\end{theorem}

\begin{proof}
If Prog.~(\ref{eq:pop}) is feasible, let $\bm{a}_0$ be one feasible solution, then we know $\gamma-l(\bm{a}_0,\bm{x})>0$ over $\bm{x}\in \mathcal{K}$ for any $\gamma>0$,. 
By \cref{thm:putinar}, there exist SOS polynomials $\sigma_i$ for $0\le i\le m$ such that the constraint in Prog.~(\ref{eq:sos}) holds for any $\gamma>0$. 
Hence, the infimum $\gamma^*\le 0$.
\end{proof}
}

In~\cite{parrilo2000}, Parrilo showed that Prog.~(\ref{eq:sos}) can be approximated by solving a series of its relaxations.
The idea is to impose restrictions on the highest degree of involved polynomials in the constraints.
Concretely, 
given a relaxation order $d_r\in \mathbb{N}$ with $2d_r\ge \max\{d, \deg(p_1),\dots, \deg(p_m)\}$, we set the degrees of the unknown SOS polynomials appropriately such that the maximum degree of polynomials involved in Prog.~(\ref{eq:sos}) equals $2d_r$.
The resulting program is referred to as the $d_r$-th SOS relaxation of Prog.~(\ref{eq:sos}):
\begin{equation}\label{eq:sos-relax}
    \begin{aligned}
    \min &\quad \gamma \\
    \textit{s.t.} & \quad \gamma -l(\bm{a},\bm{x}) = \sigma_0(\bm{x}) +\sum_{i=1}^m \sigma_i(\bm{x})\cdot p_i(\bm{x}),\\
    &\quad \sigma_0\in \Sigma^{2d_r}[\bm{x}],
    \sigma_i \in \Sigma^{2\lfloor \frac{2d_r-\deg(p_i)}{2} \rfloor}[\bm{x}], \quad \text{for } i=1,\dots,m,
    \end{aligned}
\end{equation}
where $\lfloor \cdot \rfloor$ returns the largest integer less than or equal to the argument.
As the relaxation order~$d_r$ increases and approaches infinity, the optimal value of Prog.~(\ref{eq:sos-relax}) converges to the optimal value of Prog.~(\ref{eq:sos})~\cite{Lasserre01}. 
In other words, the series of SOS relaxations yields progressively more accurate approximations of the original problem Prog.~(\ref{eq:pop}). 
One distinct advantage of Prog.~(\ref{eq:sos-relax}) is that it can be solved as an SDP.

\revise{
\begin{remark}\label{remark:inf}
Whether the infimum $\gamma^*$ is attainable in \cref{thm:pre-sound} depends on whether the solutions of the SOS relaxations Prog.~(\ref{eq:sos-relax}) converge in finitely many steps.
Hence, this is also referred to as the \emph{finite convergence property}.
\cite[Thm.~1.1]{nie2014optimality} shows that this property is decidable by checking that the constraints are not in some pathological forms, which generally holds.
In practice, since we only deal with SOS relaxations, we do not need to consider this problem and treat it as a technical assumption.
\end{remark}
}

\paragraph{SDP Translation}

Let $S_n$ denote the set of symmetric $n\times n$ real matrices.
A matrix $X\in S_n$ is \emph{positive semidefinite} if all its eigenvalues are nonnegative, denoted by $X\succeq \bm{0}$.

\begin{definition}[Standard Form SDP~{\cite[Sec.~4.6.2]{boyd2004convex}}]\label{def:sdp}
A standard form SDP has linear equality constraints and a matrix nonnegativity constraint on the variable $X\in S_n$:
\begin{equation}
    \begin{aligned}
        \min \quad &  \textbf{tr}(CX)\\
        s.t. \quad & \textbf{tr}(A_i X) = b_i, \quad i=1,\dots,k\\
        & X\succeq \bm{0}.
    \end{aligned}
\end{equation}
where $C,A_1,\dots, A_k\in S_n$ for some $k\in \mathbb{N}$, and $\textbf{tr}(\cdot)$ is the trace operator, i.e., $\textbf{tr}(CX)= \sum_{i,j=1}^n C_{ij}X_{ji}$.
\end{definition}

Let $\mm_d(\bm{x})$ be a column vector with all monomials in $\bm{x}\in \mathbb{R}^n$ of degree up to $d$.
For example, when $\bm{x} = (x_1,x_2)$, $\mm_2(\bm{x})=(1,x_1,x_2,x_1^2,x_1x_2,x_2^2)$.
Any polynomial $p(\bm{x})\in \mathbb{R}^{2d}[\bm{x}]$ can be represented by 
\begin{equation*}
    p(\bm{x}) = \mm_d^\top(\bm{x}) C_p \mm_d(\bm{x}),
\end{equation*}
where $\mm_d^\top(\bm{x})$ is the transpose of $\mm_d(\bm{x})$ and $C_p\in S_{\binom{n+d}{d}}$ is called the \emph{Gram matrix} of $p$. 
An important theorem states that a polynomial $p$ is an SOS polynomial if and only if its Gram matrix is positive semidefinite, i.e., $C_p \succeq \bm{0}$~\cite[Prop.~2.1]{lasserre2009book}. 

For each relaxation order $d_r\in \Nat$, Prog.~(\ref{eq:sos-relax}) can be translated into an equivalent SDP:
\begin{equation}\label{eq:sos-relax-sdp}
    \begin{aligned}
        \min \quad & \gamma\\
        s.t. \quad & C_{\gamma-l(\bm{a},\bm{x})} = C_{\sigma_0} + \sum_{i=1}^m C_{\sigma_i \cdot p_i}\\
        & \text{diag}(C_{\sigma_0}, C_{\sigma_1},\dots,C_{\sigma_m})\succeq \bm{0}
    \end{aligned}
\end{equation}
where 
$C_{\sigma_i \cdot p_i} = C_{\sigma_i} \mm_{d_r}(\bm{x}) \mm^\top_{d_r}(\bm{x}) C_{p_i}$ for $i=1,\dots,m$,
and $\text{diag}(C_{\sigma_0}, C_{\sigma_1},\dots,C_{\sigma_m})$ is a block-diagonal matrix.
One can check that Prog.~(\ref{eq:sos-relax-sdp}) conforms to the standard form of SDP in \cref{def:sdp}.

Roughly speaking, the complexity for solving Prog.~(\ref{eq:sos-relax-sdp}) depends on the maximum size of Gram matrices $C_{\sigma_i}$ , which is at most $\binom{n+d_r}{n}\times \binom{n+d_r}{n}$~\cite{Roux16}. 
Note that $\binom{n+d_r}{n}$ is a polynomial in $n$ for a fixed $d_r$ and vice versa, but is not a polynomial in both $n$ and $d_r$.
Since an SDP can be solved in polynomial time, for example, by interior point methods~\cite{boyd2004convex}, the complexity for solving Prog.~(\ref{eq:sos-relax}) is polynomial in $\binom{n+d_r}{n}$.

\subsection{Problem Formulation}\label{sec:2-2}

\paragraph{Program Model}
In this paper, we focus on synthesizing invariants for loops of the form in Code~\ref{code:model}. 
In \cref{sec:extension-2}, we will discuss how to handle nested loops.

\begin{listing}[!ht]
\begin{minted}[mathescape, escapeinside=||]{c}
    // Program variables: $\bm{x}\in \mathbb{R}^n$
    // Precondition: $\mathit{Pre} = \{ \bm{x}\mid \bm{q}_{pre}(\bm{x}) \le 0\}$ 
    while (|$\bm{g}(\bm{x})\le 0$|) {
        case (|$\bm{c}_1(\bm{x}) \le 0$|) : |$\bm{x} \gets \bm{f}_1(\bm{x})$|;
        case (|$\bm{c}_2(\bm{x}) \le 0$|) : |$\bm{x} \gets \bm{f}_2(\bm{x})$|;
        |$\cdots$|
        case (|$\bm{c}_k(\bm{x}) \le 0$|) : |$\bm{x} \gets \bm{f}_k(\bm{x})$|;
    }
    // Postcondition: $\mathit{Post} = \{ \bm{x} \mid \bm{q}_{post}(\bm{x})\le 0\}$
\end{minted}
\caption{The Program Model}
\label{code:model}
\end{listing}

In our program model, program variables $\bm{x}\in \mathbb{R}^n$ are assumed to take real values.
The loop consists of a loop guard $\bm{g}(\bm{x})\le \bm{0}$ and a switch-case loop body, where each branch contains a branch conditional $\bm{c}_i(\bm{x})\le \bm{0}$ and an assignment statement $\bm{x}=\bm{f}_i(\bm{x})$ for $i=1,\dots,k$.
Here, we require that $\bm{g}(\bm{x}), \bm{c}_i(\bm{x}), \bm{f}_i(\bm{x})$ are all (vectors of) polynomials.
The branch conditionals $\bm{c}_i(\bm{x})\le \bm{0}$ are tested in parallel.
This means that if more than one branch conditionals are satisfied, the program will \emph{nondeterministically} choose a satisfied branch.

The goal is to prove the correctness of the program, 
i.e., for any state satisfying the precondition ($\bm{x}\in \mathit{Pre}$), if the loop terminates, the final state must satisfy the postcondition ($\bm{x}\in\mathit{Post}$). 
Here $\mathit{Pre}$ and $\mathit{Post}$ are basic semialgebraic sets defined by polynomial inequalities $\bm{q}_{pre}(\bm{x}) \le \bm{0}$ and $\bm{q}_{post}(\bm{x})\le \bm{0}$, respectively.

\revise{
We make one assumption in our model:
Throughout the execution of the program, the program state $\bm{x}$ remains within a known hyper-rectangle $C_{\bm{x}}\subseteq \mathbb R^n$.
In our algorithms, we consider $C_{\bm{x}}$ to be of the form $\{\bm{x}\in \mathbb R^n\mid x_1^2 - N^2 \le 0, \dots, x_n^2 - N^2\le 0\}=[-N,N]^n$, where $N\in \mathbb{N}$ is a constant.

This assumption corresponds to the Archimedean condition in \cref{thm:putinar}. 
In most cases, many real-world programs have natural bounds for program variables.
Additionally, in practical programming languages like \texttt{C}, variables are typically assigned types and have known value ranges.
Therefore, the assumption is often reasonable.
As discussed in Remark~\ref{remark:homo}, this assumption is not essential and can be removed if we use extensions of Putinar's Positivstellensatz.
Moreover, our algorithms remain sound even without this assumption, which is similar to other works based on Putinar's Positivstellensatz (\cref{thm:putinar})~\cite{AGM15SAS,chatterjee2020polynomial,GHM+23OOPSLA}.
}

\paragraph{Invariant Synthesis Problem}
The formal definition of loop invariants is formulated as follows:

\begin{definition}[Invariant]
\label{def:inv}
    $\mathit{Inv} \subseteq \mathbb{R}^n$ is an invariant of the program in Code~\ref{code:model} if it satisfies the following three conditions, also called the invariant conditions:
    \begin{align}
        \bm{x} \in \mathit{Pre} &\implies \bm{x} \in \mathit{Inv}, \tag{Initial Cond.}\\
        \bm{x} \in \mathit{Inv} \wedge \bm{g}(\bm{x})\le \bm{0} \wedge \bm{c}_i(\bm{x}) \leq \bm{0} &\implies 
            \bm{f}_i(\bm{x}) \in  \mathit{Inv}, \quad \text{ for }i=1, \dots, k \tag{Inductive Cond.}\\
        \bm{x} \in \mathit{Inv} \wedge \neg (\bm{g}(\bm{x})\le \bm{0}) &\implies \bm{x} \in \mathit{Post} \tag{Saturation Cond.}
    \end{align}    
\end{definition}

The existence of an invariant implies the correctness of the loop. 
However, directly searching for a satisfactory  $\mathit{Inv}$ within the entire space of all subsets of $\mathbb{R}^n$ could be challenging.
To address this issue, 
one common approach is to impose constraints on the invariants $\mathit{Inv}$ to adhere to specific types of parametric formulas.
For explanation, we primarily focus on polynomial templates, which are defined as follows.
The extension to basic semialgebraic templates (defined in \cref{sec:extension-3}) is straightforward.

\begin{definition}[Polynomial Template]
\label{def:template}
    \revise{
    A polynomial template is a polynomial $I(\bm{a}, \bm{x}) \in \mathbb{R}[\bm{a},\bm{x}]$ defined over $C_{\bm{a}} \times C_{\bm{x}}$,  
    where $C_{\bm{a}}\subseteq \mathbb{R}^{n'}$ is a hyper-rectangle and  
    $\bm{a} = (a_1, a_2, \dots, a_{n'}) \in C_{\bm{a}}$ 
    are referred to as parameters. 
    } 
    Given a parameter assignment $\bm{a}_0 \in C_{\bm{a}}$, 
    the instantiation of the invariant $\mathit{Inv}$ w.r.t. $\bm{a}_0$ is 
    the set ${\{ \bm{x}\in C_{\bm{x}} \mid I(\bm{a}_0, \bm{x}) \leq 0 \}}$,
    where $C_{\bm{x}}=[-N,N]^n$ for some user-defined $N\in \mathbb N$.
\end{definition}

The reason for the assumption $\bm{a}\in C_{\bm{a}}$ is similar to that of $\bm{x}$.
However, when $I(\bm{a},\bm{x})$ is linear in~$\bm{a}$ as in \cref{eq:template}, we can take $C_{\bm{a}}$ to be $[-1, 1]^{n'}$ without loss of generality.
This is because the parameters $\bm{a}$ can be scaled by any positive constant without changing the invariant candidate they define.

When a polynomial template $I(\bm{a},\bm{x})$ is fixed, the invariant conditions can be expressed as constraints in first-order logic: 
\begin{align}
    &\forall \bm{x} \in C_{\bm{x}}.~\bm{q}_{pre}(\bm{x})\le \bm{0}\implies I(\bm{a},\bm{x})\le \bm{0},\label{eq:inv-pre}\\
    &\forall \bm{x} \in C_{\bm{x}}.~I(\bm{a},\bm{x})\le \bm{0}\wedge \bm{g}(\bm{x})\le \bm{0} \wedge \bm{c}_i(\bm{x})\le \bm{0} \implies I(\bm{a},\bm{f}_i(\bm{x}))\le \bm{0}, \quad i=1,\dots,k, \label{eq:inv-branch}\\
    &\forall \bm{x} \in C_{\bm{x}}.~I(\bm{a},\bm{x})\le \bm{0} \wedge \neg(\bm{g}(\bm{x})\le \bm{0}) \implies \bm{q}_{post}(\bm{x})\le \bm{0}.\label{eq:inv-post}
\end{align}

\begin{definition}[Valid and Valid Set]
\label{def:valid}
    Given a program as presented in Code~\ref{code:model} and a polynomial template $I(\bm{a},\bm{x})\in \mathbb{R}[\bm{a},\bm{x}]$, 
    a parameter assignment $\bm{a}_0 \in C_{\bm{a}}$ is valid if it satisfies constraints~(\ref{eq:inv-pre})-(\ref{eq:inv-post}), meaning that the set $\{ \bm{x} \mid I(\bm{a}_0,\bm{x}) \leq 0 \}$ is an invariant of the program. 
    The valid set, denoted by $R_I$, represents the collection of all valid parameter assignments for the polynomial template $I(\bm{a},\bm{x})$. 
\end{definition}

Given a polynomial template for Code.~\ref{code:model}, we are interested in two problems: 
\begin{enumerate}
    \item The \textbf{weak invariant synthesis problem} asks for an invariant satisfying the template, i.e., finding a valid parameter assignment $\bm{a}_0\in C_{\bm{a}}$.
    \item The \textbf{strong invariant synthesis problem}~\cite{chatterjee2020polynomial} asks for a characterization of all possible invariants satisfying the template, i.e., characterizing the valid set $R_I$.
\end{enumerate}

In the following remark, we briefly explain the difficulty in applying \cref{thm:putinar} to the above problems.

\begin{remark}\label{remark:bmi}
Let $I(\bm{a},\bm{x})$ be a polynomial template that is linear in the parameters $\bm{a}$.
For the invariant condition \cref{eq:inv-pre}, since the parameters $\bm{a}$ occur after the implication symbol, the SOS relaxations of \cref{eq:inv-pre} can still be translated to SDP, similar to Prog.~(\ref{eq:pop}).
However, \cref{eq:inv-branch} and \cref{eq:inv-post} cannot be handled in the same manner as the parameters $\bm{a}$ occur before the implication symbol.
In other words, solving \cref{eq:inv-branch} and \cref{eq:inv-post} requires one to solve a program of the following form:
\begin{equation}\label{eq:para-pop}
    \begin{aligned}
        \text{find} \quad & \bm{a}\\
        s.t. \quad &
        \forall \bm{x}. 
        \bigwedge_{i=1}^m p_i(\bm{a},\bm{x}) \ge 0\implies l(\bm{a},\bm{x}) \le 0,    
    \end{aligned}
\end{equation}
where $p_i(\bm{a},\bm{x}), l(\bm{a},\bm{x})\in \Real[\bm{a},\bm{x}]$.

\revise{
Obviously, Prog.~(\ref{eq:para-pop}) is a generalization of Prog.~(\ref{eq:pop}), allowing polynomials $p_i$ to contain unknown parameters $\bm{a}$.
Unfortunately, Prog.~(\ref{eq:para-pop}) is much harder to solve.
}
To see this, we apply \cref{thm:putinar} to transform it into a program involving SOS polynomials, assuming the Archimedean condition:
\begin{equation}\label{eq:para-sos}
\begin{aligned}
    \min &\quad \gamma \\
    \textit{s.t.} &\quad \gamma -l(\bm{a},\bm{x})= \sigma_0(\bm{x}) +\sum_{i=1}^m \sigma_i(\bm{x}) \cdot p_i(\bm{a},\bm{x}),\\
    &\quad \sigma_0 \in \Sigma[\bm{x}], \sigma_i \in \Sigma[\bm{x}], \quad \text{for } i=1,\dots, m,
\end{aligned}
\end{equation}
where the decision variables are parameters $\bm{a}$ and the unknown coefficients in SOS polynomials $\sigma_i$, for $i=0,\dots,m$.
Similarly, by restricting the highest degree of involved polynomials in constraints, we obtain a series of SOS relaxations of the above program.

\revise{
However, the resulting SOS relaxations of \cref{eq:para-sos} can not be translated into SDPs because the Gram matrices of the products $\sigma_i(\bm{x}) \cdot p_i(\bm{a},\bm{x})$ contain bilinear terms arising from the product of unknown coefficients in $\sigma_i$ and parameters $\bm{a}$.
}
These constraints, known as bilinear matrix inequalities (BMIs) in optimization theory, are incompatible with the linear matrix inequalities (LMIs) allowed in SDPs. 
As shown in~\cite{toker95bmi} and~\cite{blondel95}, solving general BMI optimization problems is \textbf{NP}-hard.
In~\cite{chatterjee2020polynomial} and~\cite{GHM+23OOPSLA}, the constraints are further encoded into quadratic programming by applying Cholesky decomposition~\cite{golub96book} to Gram matrices of SOS polynomials. However, solving non-convex quadratic programming is still \textbf{NP}-hard~\cite{sahni74}.
\end{remark}

\section{Synthesizing Strong Invariants From Polynomial Templates}
\label{sec:under-approximation}

In this section, we propose the \textbf{Cluster algorithm} to give an approximate solution to the strong invariant synthesis problem. 
To this end, we leverage Lasserre's technique from~\cite{lasserre15} to construct a series of SOS relaxations for under-approximating the valid set $R_I$.
By solving these relaxations as SDPs, we obtain a sequence of polynomials $h_{d}(\bm{a})$ for $d\in \mathbb N$ whose $0$-sublevel sets are subsets of $R_I$.
Furthermore, we show that these under-approximations possess desirable properties, including soundness, convergence, and semi-completeness.
Moreover, these under-approximations can be utilized to simplify the weak invariant synthesis problem, which reduces to finding a solution to $h_{d}(\bm{a})\le 0$ over $\bm{a}\in C_{\bm{a}}$.
 
\subsection{Lasserre's Technique}
\label{sec:3-1}
In this part, we demonstrate how to apply Lasserre's technique~\cite{lasserre15} to deal with Prog.~(\ref{eq:para-pop}).
Instead of attempting to find a valid assignment of $\bm{a}\in C_{\bm{a}}$, our objective is to under-approximate the set of valid assignments of $\bm{a}$, denoted by
\begin{equation} \label{eq:R-def}
    R = \set[\bigg]{ \bm{a}\in C_{\bm{a}} \given 
    \forall \bm{x}. \bigwedge_{i=1}^m p_i(\bm{a},\bm{x}) \ge 0\implies l(\bm{a},\bm{x}) \le 0},
\end{equation}
where $p_i(\bm{a},\bm{x}), l(\bm{a},\bm{x})\in \mathbb R[\bm{a},\bm{x}]$.
We assume that variables $\bm{x}$ and parameters $\bm{a}$ are both bounded within some predefined hyper-rectangles $C_{\bm{x}}$ and $C_{\bm{a}}$, respectively.
Without loss of generality, we can further assume that $p_i(\bm{a},\bm{x})$ include those polynomials that define $C_{\bm{x}}$ and $C_{\bm{a}}$, ensuring that $\qm (p_1,\dots,p_m)$ is Archimedean.

Let $\mathcal{K}_{\bm{a}} = \{ \bm{x}\in C_{\bm{x}}\mid \bigwedge_{i=1}^m p_i(\bm{a},\bm{x})\ge 0 \}$ be a basic closed semialgebraic set parameterized by $\bm{a}$.
When $\mathcal{K}_{\bm{a}}$ is non-empty for every $\bm{a}\in C_{\bm{a}}$,    
we can express $R$ as 
\begin{equation}\label{eq:R}
    R = \{ \bm{a}\in C_{\bm{a}} \mid J(\bm{a}) \le 0\},
\end{equation}
where $J(\bm{a}) = \sup_{\bm{x} \in \mathcal{K}_{\bm{a}}} l(\bm{a}, \bm{x})$.

Therefore, if we can find a function $h(\bm{a})$ such that $h(\bm{a})\ge J(\bm{a})$ for all $\bm{a}\in C_{\bm{a}}$, then the $0$-sublevel set of $h(\bm{a})$ serves as an under-approximation of $R$, i.e.,
\begin{equation}
    \{\bm{a}\in C_{\bm{a}} \mid h(\bm{a})\le 0\}
    \subseteq 
    \{\bm{a}\in C_{\bm{a}} \mid J(\bm{a})\le 0\} = R.
\end{equation}
Now the problem boils down to finding such a function $h(\bm{a})$. 
Ideally, we would like $h(\bm{a})$ to be a simple expression that closely approximates $J(\bm{a})$.
Fortunately, according to the following \cref{lem:upper} and \cref{thm:upper}, we can restrict our search to polynomials for $h(\bm{a})$.

\begin{lemma}~\cite[Lem.~1]{lasserre15}\label{lem:upper}
    The function $J$ as defined in \cref{eq:R} is upper semi-continuous, i.e., for all $\bm{a}_0 \in C_{\bm{a}}$, we have 
    \begin{equation}
        \limsup_{\bm{a} \to \bm{a}_0} J(\bm{a}) \leq J(\bm{a}_0).
    \end{equation}
\end{lemma}

\begin{theorem}~\cite[Thm.~1]{lasserre15} \label{thm:upper}
    Let $C_{\bm{a}}\subset \Real^{n'}$ be a compact set and $J(\bm{a}): C_{\bm{a}}\to \Real$ be a bounded and upper semi-continuous function. 
    Then there exists a sequence of polynomials $\{h_i(\bm{a})\mid i\in \Nat \}\subset \Real[\bm{a}]$ such that $h_i(\bm{a}) \ge J(\bm{a})$ over $\bm{a}\in C_{\bm{a}}$ for all $i\in\Nat$ and 
    \begin{equation}
        \lim_{i\to \infty} \int_{C_{\bm{a}}} |h_i(\bm{a})-J(\bm{a})| \mathrm{d} \bm{a} = 0.
    \end{equation}
\end{theorem}

In what follows, we show how to compute such a polynomial $h(\bm{a})$ by employing SOS relaxations.
First, using the definition of $J(\bm{a})$, we have $h(\bm{a}) - J(\bm{a}) \ge 0$ over $\bm{a}\in C_{\bm{a}}$ if and only if 
\begin{equation}
    \forall (\bm{a}, \bm{x}) \in C_{\bm{a}}\times \mathcal{K}_{\bm{a}}. ~h(\bm{a}) - l(\bm{a},\bm{x}) \ge 0.
\end{equation}
Thus, the problem amounts to solving the following program:
\begin{equation}\label{eq:qsos-prog}
\begin{aligned}
    \inf \quad & \frac{1}{\mu(C_{\bm{a}})}\int_{C_{\bm{a}}} h(\bm{a}) \mathrm{d} \bm{a}\\
    s.t. \quad & \forall (\bm{a}, \bm{x}).~\bigwedge_{i=1}^m p_i(\bm{a},\bm{x}) \ge 0 \implies h(\bm{a}) - l(\bm{a},\bm{x})\ge 0,
\end{aligned}    
\end{equation}
where $\mu(C_{\bm{a}})$ is the volume of $C_{\bm{a}}$ and the objective function is the scaled integral of $h(\bm{a})$ over $C_{a}$.
Since $J(\bm{a})$ is fixed, minimizing $\int_{C_{\bm{a}}} h(\bm{a}) \mathrm{d} \bm{a}$ is the same as minimizing $\int_{C_{\bm{a}}} |h(\bm{a}) - J(\bm{a})| \mathrm{d} \bm{a}$, i.e., the gap between $h(\bm{a})$ and $J(\bm{a})$ over $C_{\bm{a}}$.

The main difference between Prog.~(\ref{eq:qsos-prog}) and Prog.~(\ref{eq:para-sos}) is how parameters $\bm{a}$ are quantified in constraints:
In Prog.~(\ref{eq:para-sos}), $\bm{a}$ are associated with an (implicit) existential quantifier, while in Prog.~(\ref{eq:qsos-prog}) a universal quantifier. 
As a result, for Prog.~(\ref{eq:qsos-prog}), we can treat parameters~$\bm{a}$ equally as variables $\bm{x}$.
After applying \cref{thm:putinar}, we have:
\begin{equation}\label{eq:qsos}
\begin{aligned}
    \inf \quad & \frac{1}{\mu(C_{\bm{a}})}\int_{C_{\bm{a}}} h(\bm{a}) \mathrm{d} \bm{a} \\
    s.t. \quad & h(\bm{a}) - l(\bm{a},\bm{x}) = \sigma_0(\bm{a},\bm{x}) +  \sum_{i=1}^m \sigma_i(\bm{a},\bm{x}) \cdot p_i(\bm{a},\bm{x}),\\
    &\sigma_0(\bm{a},\bm{x}), \sigma_i(\bm{a},\bm{x})\in \Sigma[\bm{a},\bm{x}], \text{ for }i=1,\dots,m
\end{aligned}    
\end{equation}
It is worthwhile to note that SOS polynomials $\sigma_i(\bm{a},\bm{x})$ belong to $\Sigma[\bm{a},\bm{x}]$, while SOS polynomials $\sigma_i(\bm{x})\in \Sigma[\bm{x}]$ in Prog.~(\ref{eq:para-sos}).

When $h(\bm{a})$ is of degree $d$, it can be expressed as 
$h(\bm{a}) = \sum_{\bm{\beta}} \mathrm{h}_{\bm{\beta}} \bm{a}^{\bm{\beta}} $ where $\bm{\beta} = (\beta_1,\dots,\beta_{n'})\in \Nat^{n'}$ are exponents such that $\sum_{j=1}^{n'}\beta_j \le d$ and $\mathrm{h}_{\bm{\beta}}$ are unknown coefficients in $\mathbb R$.
Then, the objective function in Prog.~(\ref{eq:qsos}) can be simplified into a linear expression in coefficients $h_{\bm{\beta}}$:
\begin{equation}
    \begin{aligned}
        \frac{1}{\mu(C_{\bm{a}})} \int_{C_{\bm{a}}} h(\bm{a}) \mathrm{d} \bm{a} 
        &= \frac{1}{\mu(C_{\bm{a}})} \int_{C_{\bm{a}}} \bigg( \sum_{\bm{\beta}} \mathrm{h}_{\bm{\beta}} \bm{a}^{\bm{\beta}} \bigg) \mathrm{d} \bm{a} \\
        &= \sum_{\bm{\beta}} \bigg(\frac{1}{\mu(C_{\bm{a}})} \int_{C_{\bm{a}}} \bm{a}^{\bm{\beta}} \mathrm{d} \bm{a} \bigg) \mathrm{h}_{\bm{\beta}}.
    \end{aligned}
\end{equation} 
When $C_{\bm{a}}$ is a hyper-rectangle (or other simple shapes like ellipses), the integral $\int_{C_{\bm{a}}} \bm{a}^{\bm{\beta}} \mathrm{d} \bm{a}$ will be easy to compute.

Now we present SOS relaxations of Prog.~(\ref{eq:qsos}). 
Assuming that $h(\bm{a})$ is of degree~$d$, 
let $d_r$ be the smallest natural number such that  $2d_r\ge \max\{d, \deg(l), \deg(p_1),\dots, \deg(p_m)\}$, then the $d_r$-th SOS relaxation of Prog.~(\ref{eq:qsos}) is given by 
\begin{equation}\label{eq:qsos-relax}
\begin{aligned}
        \min \quad &\sum_{\bm{\beta}} \gamma_{\bm{\beta}} \mathrm{h}_{\bm{\beta}} \\
    s.t. \quad 
    & h(\bm{a}) = \sum_{\bm{\beta}} \mathrm{h}_{\bm{\beta}} \bm{a}^{\bm{\beta}} \in \Real^d[\bm{a}],\\
    &h(\bm{a}) - l(\bm{a},\bm{x}) = \sigma_0(\bm{a},\bm{x}) +  \sum_{i=1}^m \sigma_i(\bm{a},\bm{x}) \cdot p_i(\bm{a},\bm{x}),\\
    &\sigma_0\in \Sigma^{2d_r}[\bm{x}],
    \sigma_i \in \Sigma^{2\lfloor \frac{2d_r-\deg(p_i)}{2} \rfloor}[\bm{x}], \quad \text{for } i=1,\dots,m,
\end{aligned}    
\end{equation}
where $\gamma_{\bm{\beta}} = \frac{1}{\mu(C_{\bm{a}})}\int_{C_{\bm{a}}} \bm{a}^{\bm{\beta}} \mathrm{d} \mu(\bm{a})$.
 
For each $d\in \Nat$, mirroring the process in \cref{sec:pre-sos}, Prog.~(\ref{eq:qsos-relax}) can be translated into an SDP.
\revise{
If this SDP is solvable, we can use the solutions $\{h_{\bm{\beta}}\}_{\bm{\beta}}$ to construct a polynomial $h_d(\bm{a})= \sum_{\bm{\beta}} \mathrm{h}_{\bm{\beta}} \bm{a}^{\bm{\beta}}$, which serves as the $d$th approximation of $h(\bm{a})$.
}
If the translated SDP is not solvable, we set $h_d(\bm{a})=1$.
In this case, $R_d=\{\bm{a}\in C_{\bm{a}} \mid 1\le 0\}=\emptyset$ is a trivial under-approximation of $R$.
Moreover, we have the following theorem:
\begin{theorem} 
\label{thm:lasserre15}
\cite[Thm.~5]{lasserre15}
    Assume that $R$ has nonempty interior and $\mathcal K_{\bm{a}}$ is non-empty for every $\bm{a}\in C_{\bm{a}}$,
    then $h_{d}(\bm{a})$ converges to $J(\bm{a})$ (from above) as $d$ goes to $\infty$, i.e.,
    \begin{equation*}
        \lim_{d\to \infty} \int_{C_{\bm{a}}} |h_{d}(\bm{a})-J(\bm{a})| \mathrm{d} \bm{a} = 0.
    \end{equation*}
\end{theorem}

\subsection{Cluster Algorithm}
In this part, we show how to apply Lasserre's technique to solve the strong invariant synthesis problem.

Similar to \cref{eq:R}, we first express the valid set $R_I$ as a $0$-sublevel set.
Mimicking the definition of $J$, 
let us define $J_{1}(\bm{a}), \ldots, J_{k+2}(\bm{a})$ as follows (recall that $k$ is the number of branches):
\begin{align}
    J_{1}(\bm{a}) =& \sup_{\bm{x}\in \mathcal{K}_{\bm{a},1}} I(\bm{a},\bm{x}),
    \text{ with } \notag\\
    & \quad \mathcal{K}_{\bm{a},1} = \{ \bm{x}\in C_{\bm{x}} \mid \bm{q}_{pre}(\bm{x})\le \bm{0} \}, \label{eq:J-pre}\\
    J_{i+1}(\bm{a}) =& \sup_{\bm{x}\in \mathcal{K}_{\bm{a},i+1}} I(\bm{a},\bm{f}_i(\bm{x})),  
    \text{ with } \notag \\
    & \quad \mathcal{K}_{\bm{a},i+1} = \{ \bm{x}\in C_{\bm{x}} \mid I(\bm{a},\bm{x})\le \bm{0}, \bm{g}(\bm{x})\le \bm{0}, \bm{c}_i(\bm{x})\le \bm{0} \}, \text{ for } i=1,\dots,k, \label{eq:J-branch}\\
    J_{k+2}(\bm{a}) =& \sup_{\bm{x}\in \mathcal{K}_{\bm{a},k+2}} \bm{q}_{post}(\bm{x}), \text{ with } \notag \\
    &\quad \mathcal{K}_{\bm{a},k+2} = \{ \bm{x}\in C_{\bm{x}} \mid I(\bm{a},\bm{x})\le \bm{0} \wedge \neg(\bm{g}(\bm{x})\le \bm{0}) \}. \label{eq:J-post}
\end{align}
It is straightforward to see that Eqs.~(\ref{eq:J-pre})-(\ref{eq:J-post}) correspond to Eqs.~(\ref{eq:inv-pre})-(\ref{eq:inv-post}), respectively.
Then, we define 
\begin{equation}\label{eq:J}
    J(\bm{a}) = \max \{J_1(\bm{a}),\cdots, J_{k+2}(\bm{a}), -1 \},
\end{equation}
where an addition constant $-1$ (or any other negative real number) is introduced to ensure $J(\bm{a})\neq -\infty$ in an extreme case when $\mathcal{K}_{\bm{a}, i}=\emptyset$ for $i=1,\dots,k+2$. 
By definition, the valid set $R_I$ is the $0$-sublevel set of function $J(\bm{a})$, i.e.,
\begin{equation} 
    R_I = \{ \bm{a} \in C_{\bm{a}} \mid J(\bm{a}) \leq 0 \}.
\end{equation}

\revise{
In order to obtain a close under-approximation of $R_I$, we try to find a polynomial $h(\bm{a})$ that approaches $J(\bm{a})$ over $C_{\bm{a}}$ from above.
According to Prog.~(\ref{eq:qsos-prog}), the problem is reduced to the following program:
\begin{equation} \label{eq:qsos-inv}
\begin{aligned}
    &\qquad \inf \quad \frac{1}{\mu(C_{\bm{a}})}\int_{C_{\bm{a}}} h(\bm{a}) \mathrm{d} \bm{a}\\
    &\qquad s.t. \quad \forall (\bm{a}, \bm{x})\in C_{\bm{a}}\times C_{\bm{x}}. \\
    &
    \begin{cases}
        \bm{q}_{pre}(\bm{x})\le \bm{0} \implies h(\bm{a}) - I(\bm{a},\bm{x})\ge 0,\\
        I(\bm{a},\bm{x})\le \bm{0}\wedge \bm{g}(\bm{x})\le \bm{0}\wedge \bm{c}_i(\bm{x})\le \bm{0} \implies h(\bm{a}) - I(\bm{a},\bm{f}_i(\bm{x}))\ge 0,\text{ for } i=1,\dots, k,\\
        I(\bm{a},\bm{x})\le \bm{0} \wedge \neg(\bm{g}(\bm{x})\le \bm{0}) \implies h(\bm{a}) - \bm{q}_{post}(\bm{x})\ge 0,\\
        h(\bm{a}) + 1 \ge 0. 
    \end{cases}
\end{aligned}    
\end{equation}
where the first three constraints correspond to Eqs.~(\ref{eq:J-pre})-(\ref{eq:J-post}) (or the invariant conditions), and the last constraint corresponds to the additional value $-1$ in \cref{eq:J}.
}

For simplicity, we assume that $\bm{q}_{pre}$, $\bm{q}_{post}$, $\bm{g},\bm{c_i}$ are polynomials (instead of vectors of polynomials) and use $q_{pre}$, $q_{post}$, $g$, and $c_i$ instead.
We also assume that $C_{\bm{x}}=\{\bm{x}\in \mathbb R^n\mid N - x_1^2 \ge 0, \dots, N - x_n^2 \ge 0\}$ and $C_{\bm{a}}=\{\bm{x}\in \mathbb R^{n'}\mid 1- a_1^2 \ge 0, \dots, 1 - a_{n'}^2 \ge 0\}$.
Let polynomial $h(\bm{a})$ be of degree $d$, we translate Prog.~(\ref{eq:qsos-inv}) into constraints with SOS polynomials: 
\begin{equation}\label{eq:qsos-inv-relax}
\begin{aligned} 
        \inf \quad &\sum_{\bm{\beta}} \gamma_{\bm{\beta}} \mathrm{h}_{\bm{\beta}}\\
        s.t.  \quad
        & h(\bm{a}) = \sum_{\bm{\beta}} \mathrm{h}_{\bm{\beta}} \bm{a}^{\bm{\beta}} \in \Real^d[\bm{a}],\\
        & h(\bm{a}) - I(\bm{a},\bm{x})  = \sigma_{0,0} - \sigma_{0,1} \cdot q_{pre}(\bm{a},\bm{x}) + \sum_{j=1}^n \sigma^{\bm{x}}_{0,j} \cdot (N-x_j^2) + \sum_{j=1}^{n'} \sigma^{\bm{a}}_{0,j} \cdot (1-a_j^2),\\
        & h(\bm{a}) - I(\bm{a},f_i(\bm{x})) 
        = \sigma_{i,0} - \sigma_{i,1}\cdot I(\bm{a},\bm{x})-  \sigma_{i,2} \cdot g(\bm{x}) - \sigma_{i,3} \cdot c_i(\bm{x})\\
        &\hspace*{10em} + \sum_{j=1}^n \sigma^{\bm{x}}_{i,j} \cdot (N-x_j^2) + \sum_{j=1}^{n'} \sigma^{\bm{a}}_{i,j} \cdot (1-a_j^2) ,\text{ for } i=1,\dots,k  \\
        & h(\bm{a}) - q_{post}(\bm{x}) 
        = \sigma_{k+1,0}  - \sigma_{k+1,1} \cdot I(\bm{a},\bm{x}) + \sigma_{k+1,2}\cdot g(\bm{x}) \\
        &\hspace*{10em} + \sum_{j=1}^n \sigma^{\bm{x}}_{k+1,j} \cdot (N-x_j^2) + \sum_{j=1}^{n'} \sigma^{\bm{a}}_{k+1,j} \cdot (1-a_j^2),\\
        &h(\bm{a}) + 1 
         = \sigma_{k+2,0}+\sum_{j=1}^n \sigma^{\bm{x}}_{k+2,j} \cdot (N-x_j^2) + \sum_{j=1}^{n'} \sigma^{\bm{a}}_{k+2,j}\cdot (1-a_j^2),\\
        & \sigma_{i,j}, \sigma^{\bm{x}}_{i,j}, \sigma_{i,j}^{\bm{a}}\in \Sigma[\bm{a},\bm{x}] \text{ for all pair $(i,j)$}, 
\end{aligned}
\end{equation}
where, according to the definition of $C_{\bm{a}}$,
\begin{equation}
    \gamma_{\bm{\beta}} = \frac{1}{2^{n'}}\int_{C_{\bm{a}}} \bm{a}^{\bm{\beta}} \mathrm{d} \bm{a}=
    \begin{cases}
        0  &\text{if $\beta_i$ is odd for some $i$},\\
        \prod_{i=1}^n (\beta_i+1)^{-1} & \text{otherwise.}
    \end{cases}
\end{equation}

\revise{
The Cluster algorithm takes as input a program of the form Code~\ref{code:model}, a polynomial template $I(\bm{a},\bm{x})\in \Real[\bm{a},\bm{x}]$, and an upper bound $D\in \mathbb{N}$ on the degree of $h(\bm{a})$.
For each degree $d$ such that $1\le d \le D$,} the algorithm tries to find a polynomial $h_d(\bm{a})$ of degree $d$ by solving the $d_r$th SOS relaxation of Prog.~(\ref{eq:qsos-inv-relax}), where $d_r$ is the smallest natural number such that $2d_r$ is larger than or equal to the maximum degree of polynomials occurring in Prog.~(\ref{eq:qsos-inv}).
If the program is solvable, an under-approximation of $R_I$ is given by
\begin{equation}
    R_{I,d}=\{\bm{a}\in C_{\bm{a}}\mid h_{d}(\bm{a})\le 0\}\subseteq R_I.   
\end{equation}
If not solvable, we set $h_d(\bm{a})=1$ and $R_{I,d}=\{\bm{a}\in C_{\bm{a}}\mid 1\le 0\}=\emptyset$.
\revise{
Finally, the algorithm outputs the sequence $\{R_{I,d}\mid 1\le d\le D\}$.
The pseudo code of the algorithm is presented in \cref{alg:cluster}. 
}

\revise{
The Cluster algorithm can also be adapted to solve the weak invariant synthesis problem.}
When $h_d(\bm{a})$ is obtained for some $d\in \Nat$, we know that any parameter assignment $\bm{a}\in R_{I,d}$ is valid. 
From this perspective, the polynomial $h_d(\bm{a})$ characterizes a \emph{cluster} of invariants of similar shapes. 
To synthesize a valid assignment $\bm{a}$ such that $h(\bm{a})\le 0$, we only need to solve the constraint $\exists \bm{a}\in C_{\bm{a}}.~h(\bm{a})\le 0$ \revise{(see the comment on line 13 in \cref{alg:cluster})}, which is usually simpler than the original invariant conditions Eqs.~(\ref{eq:inv-pre})-(\ref{eq:inv-post}) and can be tackled by many modern optimization tools or SMT solvers.

\revise{
\begin{remark}\label{remark:parameter}
    One minor (theoretical) advantage of the Cluster algorithm is that it can handle templates $I(\bm{a},\bm{x})$ nonlinear in parameters~$\bm{a}$, as we only assume that $I(\bm{a},\bm{x}) \in \Real[\bm{a},\bm{x}]$.
    This is not the case for many existing approaches~\cite{chatterjee2020polynomial,GHM+23OOPSLA,LWY+14FCS,AGM15SAS} (discussed in \cref{sec:related}),
    where the templates are required to contain linear parameters so that the constraints can be encoded into desired forms.
    However, since in practice parameters often represent unknown coefficients, we still mainly focus on the case when templates are linear in parameters.
\end{remark}
}

\IncMargin{1em}
\begin{algorithm2e}[t]
\SetKwInOut{Input}{Input}
\SetKwInOut{Output}{Output}
\ResetInOut{Output} 
\Input{A program $\mathcal{P}$ of the form Code~\ref{code:model}, a polynomial template $I(\bm{a},\bm{x})\in \Real[\bm{a},\bm{x}]$, and an upper bound $D\in \mathbb{N}$ on the degree of $h(\bm{a})$.}
\Output{A sequence of under-approximations $\{ R_{I,i} \mid 1\le i\le D\}$.}
\BlankLine
Construct Prog.~(\ref{eq:qsos-inv-relax}) using $\mathcal{P}$ and $I(\bm{a},\bm{x})$\;
$d \gets 1$\;
\While{
    $d \le D$
}{
$d_{\max} \gets$ the largest degree of polynomials in Prog.~(\ref{eq:qsos-inv-relax})\;
$d_r \gets \lfloor \frac{d_{\max}+1}{2}\rfloor$\;
\eIf{the $d_r$-th SOS relaxation of Prog.~(\ref{eq:qsos-inv-relax}) is solvable}{
$\{h_{\bm{\beta}}\}_{\bm{\beta}}\gets$ Solve the $d_r$-th SOS relaxation of Prog.~(\ref{eq:qsos-inv-relax})\;
$h_d(\bm{a}) \gets \sum_{\bm{\beta}} \mathrm{h}_{\bm{\beta}} \bm{a}^{\bm{\beta}}$\algocomment*[l]{construct $h_d(\bm{a})$ using coefficients}
}{
$h_d(\bm{a}) \gets 1 $\algocomment*[l]{constant polynomial}
}
$R_{I,d} \gets \{\bm{a}\in C_{\bm{a}}\mid h_d(\bm{a})\le 0\}$\;
\algocomment*[h]{a valid parameter assignment can be obtained by solving $h_d(\bm{a})\le 0$}\\
$d\gets d+1$\;
}
\KwRet{$R_{I,1}, \dots, R_{I,D}$}\algocomment*[l]{a sequence of under-approximations of $R_I$}
\caption{\revise{The Cluster Algorithm}}
\label{alg:cluster}
\end{algorithm2e}
\DecMargin{1em}

\subsection{Soundness, Convergence, and Semi-Completeness}

\revise{
Now we prove the output $\{R_{I,d}\mid 1\le d\le D \}$ of the Cluster algorithm has many desired properties.
}
\revise{
\begin{theorem}[Soundness]
\label{thm:soundness}
Given $D\in \mathbb N$, $R_{I,d}$ is an under-approximation of the valid set $R_I$, i.e., $R_{I,d} \subseteq R_I$, for any $d$ such that $1\le d\le D$. 
\end{theorem}
\begin{proof}
We first show that a feasible solution to Prog.~(\ref{eq:qsos-inv-relax}) is also a solution to Prog.~(\ref{eq:qsos-inv}).
This is achieved by applying \cref{lem:qm} to each constraint in Prog.~(\ref{eq:qsos-inv-relax}).
For example, assume that there exists a polynomial $h(\bm{a})$ such that the following constraint in Prog.~(\ref{eq:qsos-inv-relax}) holds:
\begin{equation*}
    h(\bm{a}) - I(\bm{a},\bm{x})  = \sigma_{0,0} - \sigma_{0,1} \cdot q_{pre}(\bm{a},\bm{x}) + \sum_{j=1}^n \sigma^{\bm{x}}_{0,j} \cdot (N-x_j^2) + \sum_{j=1}^{n'} \sigma^{\bm{a}}_{0,j} \cdot (1-a_j^2)
\end{equation*}
for some SOS polynomials $\sigma_{0,0},\dots, \sigma^{\bm{a}}_{0,j}\in \Sigma[\bm{a},\bm{x}]$. 
By applying \cref{lem:qm}, we have 
\begin{equation*}
    \bigwedge_{j=1}^n N-x^2_j\ge 0 \wedge \bigwedge_{j=1}^{n'} N-a^2_j\ge 0 \wedge q_{pre}(\bm{a},\bm{x}) \le 0 \implies  h(\bm{a}) - I(\bm{a},\bm{x})\ge 0,
\end{equation*}
which corresponds to the first constraint in Prog.~(\ref{eq:qsos-inv}), i.e., 
\begin{equation*}
    \forall \bm{x}\in C_{\bm{x}}, \bm{a}\in C_{\bm{a}}.~q_{pre}(\bm{a},\bm{x}) \le 0 \implies  h(\bm{a}) - I(\bm{a},\bm{x})\ge 0.
\end{equation*}

Therefore, if the $d_r$th SOS relaxation of Prog.~(\ref{eq:qsos-inv-relax}) is solvable, we have $h_d(\bm{a})\ge J(\bm{a})$ over $C_{\bm{a}}$, which implies $R_{I,d} \subseteq R_I$.
If not solvable, we have $R_{I,d} = \emptyset \subseteq R_I$.
\end{proof}
}

\begin{theorem}[Convergence]
    \label{thm:convergence}
    If the set 
    $\{ \bm{a} \in C_{\bm{a}} \mid J(\bm{a}) = 0 \}$
    has Lebesgue measure zero, 
    then we have
    \begin{equation*}
    \label{eqn:convergenceRd}
        \lim_{D \to \infty} \mu (R_I \setminus R_{I,D}) = 0,
    \end{equation*}
\end{theorem}

\begin{proof}
    Essentially the same as the proof of~\cite[Thm.~3]{lasserre15}.
\end{proof}

\revise{
The condition in \cref{thm:convergence} is to ensure that 
\begin{equation}\label{eq:J-assump}
\mu (\{ \bm{a} \in C_{\bm{a}} \mid J(\bm{a}) \le 0 \})= \mu (\{ \bm{a} \in C_{\bm{a}} \mid J(\bm{a}) < 0 \}),
\end{equation}
meaning that the infimum value of Prog.~(\ref{eq:qsos-inv-relax}) is not attainable at most over a region of measure~$0$.
This assumption is made similar to the assumption in \cref{thm:pre-sound}.
As mentioned in \cref{remark:inf}, this is just a technical assumption and usually holds in practice.
}

Before presenting the semi-completeness result, we introduce the following definition.
\begin{definition}[Robustness]\label{def:robust}
    A polynomial template $I(\bm{a},\bm{x})$ is said to be \emph{robust} (w.r.t. the program model Code.~\ref{code:model}) if there exists a valid parameter assignment $\bm{a}_0\in C_{\bm{a}}$ and a small constant $\epsilon>0$ such that any $\bm{a}$ satisfying $\|\bm{a}-\bm{a}_0\|_2<\epsilon$ is still valid. 
\end{definition}

\revise{
\begin{proposition}\label{prop:decide}
Checking the robustness of a polynomial template is decidable.
\end{proposition}
\begin{proof}
By \cref{def:robust}, when the polynomial template is robust, the valid set $R_I$ contains an interior point.
Let $\varphi(\bm{a})$ denote the conjunction of formulas in Eqs.~(\ref{eq:inv-pre})-(\ref{eq:inv-post}). 
Then the problem reduces to checking whether the following first-order logic formula holds:
\begin{equation}
    \exists \epsilon > 0, \exists \bm{a}_0\in C_{\bm{a}}, \forall \bm{a}\in C_{\bm{a}}.~\|\bm{a}-\bm{a}_0\|_2<\epsilon \implies \varphi(\bm{a}),
\end{equation}
which is decidable due to Tarski's result~\cite{Tarski51}.
\end{proof}
}

\begin{theorem}[Semi-Completeness]
\label{thm:weakComplete}
    If the set $\{ \bm{a} \in C_{\bm{a}} \mid J(\bm{a}) = 0 \}$ has Lebesgue measure zero 
    and there exists a robust polynomial template $I(\bm{a}_0,\bm{x})$ for some $\bm{a}_0\in C_{\bm{a}}$,
    the Cluster algorithm will yield a \emph{non-empty} under-approximation $R_{I,D}\subseteq R_I$ for some $D\in \mathbb N$ large enough.
\end{theorem}

\begin{proof}
    By definition, a polynomial template $I(\bm{a},\bm{x})$ is robust if $R_I$ has an interior point.
    Combining \cref{thm:lasserre15} and \cref{thm:convergence},
    $R_{I,d}$ has a positive Lebesgue measure when $d$ is large enough, 
    which implies that $R_{I,d}$ is a non-empty under-approximation of the valid set $R_I$.
\end{proof}

\revise{Combining \cref{thm:convergence} and \cref{thm:weakComplete}, we know that the valid set $R_I$ admits an arbitrarily close approximation by solving Prog.~(\ref{eq:qsos-inv-relax}) for sufficiently large $d$ (and $d_r$), which gives an approximation solution to the strong invariant synthesis problem.
}
Finally, we end this section with an illustrative example.

\begin{example}
\label{ex:1}
Consider a discrete-time dynamical system presented in Code.~\ref{code:illustritive}.

\begin{listing}[!h]
\begin{minted}[mathescape, escapeinside=||]{c}
    // Program variables: $(x, y) \in \Real^2$
    // Range: $C_{x,y} = \{(x,y)\mid 4-x^2\ge 0, 4-y^2\ge 0\}$
    // Precondition: {$(x,y) \mid x^2 + y^2 - 1 \le 0$} 
    while (|$x^2-0.81\le 0$|) {
        // omit case($-1\le 0$)
        |$x \gets 0.95(x-0.1y^2)$|;
        |$y \gets 0.95(y+0.2xy)$|;
    }
    // Postcondition: {$(x,y) \mid 0.25 - x^2 - (y - 1.5)^2 \le 0$}
\end{minted}
\caption{A Discrete-Time Dynamical System}
\label{code:illustritive}
\end{listing}

Suppose that we are searching for ellipsoid-shaped invariants centered at the origin of the form 
\begin{equation}
    x^2 + \hat{a} y^2 + \hat{b} \in \Real[\hat{a},\hat{b},x,y],
\end{equation}
where $\hat{a}$ and $\hat{b}$ are parameters within the range $(\hat{a},\hat{b})\in [-10,10]^2$. 
By replacing $\hat{a}$ and $\hat{b}$ by $10a$ and $10b$, we denote the polynomial template by
\begin{equation}
    I(a,b,x,y) = x^2 + 10a y^2 + 10b,
\end{equation}
where $(a,b)\in C_{a,b}=\{(a,b)\mid 1-a^2\ge 0, 1-b^2\ge 0\} = [-1,1]^2$.

Let $D=6$ be the upper bound on the degree of $h(a,b)$. 
By applying Lasserre's technique, we solve the following program to obtain $h(a,b)$, for $1\le d\le D$:
\begin{equation}\label{eq:ex1-prog}
    \begin{aligned}
        \min \quad & \sum_{\beta_1 + \beta_2 \le d} \bm{\gamma}_{\beta_1,\beta_2}\mathrm{h}_{\beta_1,\beta_2}\\
        s.t.  \quad
        & h(a,b) = \sum_{\beta_1+\beta_2\le d} \mathrm{h}_{\beta_1,\beta_2} a^{\beta_1}b^{\beta_2} \in \Real^d[a,b],\\
        & h(a,b) - I(a,b,x,y) 
        = \sigma_{0,0} + \sigma_{0,1}\cdot (1-x^2-y^2)\\
        & \qquad  + \sigma^{a,b}_{0,1} \cdot(1-a^2)+ \sigma^{a,b}_{0,2} \cdot(1-b^2) + \sigma^{x,y}_{0,1}\cdot (4-x^2) + \sigma^{x,y}_{0,2}\cdot (4-y^2),\\
        & h(a,b) - I\big(a,b,0.95(x-0.1y^2),0.95(y+0.2xy)\big) = \sigma_{1,0} - \sigma_{1,1} \cdot I(a,b,x,y) + \sigma_{1,2}\cdot (0.81-x^2)\\
        & \qquad  + \sigma^{a,b}_{1,1} \cdot(1-a^2)+ \sigma^{a,b}_{1,2} \cdot(1-b^2) + \sigma^{x,y}_{1,1} \cdot(4-x^2) + \sigma^{x,y}_{1,2} \cdot(4-y^2),\\
        & h(a,b) - \big(0.25 - x^2 -(y-2)^2\big)
        = \sigma_{2,0} - \sigma_{2,1} \cdot I(a,b,x,y)\\
        & \qquad  + \sigma^{a,b}_{2,1} \cdot(1-a^2)+ \sigma^{a,b}_{2,2} \cdot(1-b^2) + \sigma^{x,y}_{2,1}\cdot (4-x^2) + \sigma^{x,y}_{2,2} \cdot (4-y^2),\\
        & h(a,b) + 1 = \sigma_{3,0} + \sigma^{a,b}_{3,1} \cdot (1-a^2)+ \sigma^{a,b}_{3,2} \cdot (1-b^2) + \sigma^{x,y}_{3,1} (4-x^2) + \sigma^{x,y}_{3,1}\cdot (4-y^2),\\
        & \sigma_{0,0}, \sigma_{1,0}, \sigma_{2,0}, \sigma_{3,0}  \in \Sigma^{2d_r}[a,b,x,y], 
        \sigma_{1,1}, \sigma_{2,1} \in \Sigma^{2d_r-3}[a,b,x,y],\\
        &\text{all other SOS polynomials are in } \Sigma^{2d_r-2}[a,b,x,y],
    \end{aligned}
\end{equation}
where 
\begin{equation}
    \gamma_{\beta_1,\beta_2} = \frac{1}{2^{2}}\int_{-1}^1 \int_{-1}^1 a^{\beta_1}b^{\beta_2} \mathrm{d} a \mathrm{d} b=
    \begin{cases}
        0  &\text{if either $\beta_1$ or $\beta_2$ is odd},\\
        \frac{1}{(\beta_1+1)(\beta_2+1)} & \text{otherwise,}
    \end{cases}
\end{equation}
and $d_r$ is the smallest natural number such that $2d_r\ge \max\{d, 3\}$. 
Here $\deg(I)=3$ is the largest degree of polynomials in constraints other than $h(a,b)$.

In this example, we use \textsc{Yalmip}~\cite{Lofberg2004} to formulate Prog.~(\ref{eq:ex1-prog}) and \tool{Mosek} solver~\cite{mosek} to solve the translated SDP.
For $d=1$ and $2$, the program is not solvable.
For $d=3,4,5,$ and $6$, we obtain, rounding to 5 decimal places,
\begin{align*}
    h_{3}(a,b) &= 2.69187 + \cdots - 1.50005 ab^2 - 2.10735 a^3,\\
    h_{4}(a,b) &= 2.34708 + \cdots - 0.00418 a^3b - 0.74973 a^4,\\
    h_{5}(a,b) &= 2.09144 + \cdots + 2.23803 a^4b + 2.57913 a^5,\\ 
    h_{6}(a,b) &= 1.74276 + \cdots - 1.82314 a^5b + 1.67067 a^6,
\end{align*}
which give under-approximations to the valid set $R_I$. 
In \cref{fig}, we plot the $0$-sublevel sets of the above four polynomials.
Therefore, any point $(a,b)$ in the light blue region is a valid parameter assignment for the polynomial template $I(a,b,x,y)$.

\revise{
It's important to note that for small values of $d$, the inequality $h_{d}(\bm{a})\le 0$ might not have any solutions, i.e., $R_{I,d}=\emptyset$.
When this scenario occurs, we need to increase $D$ to search for polynomials $h_{d}(\bm{a})$ of higher degrees.
Alternatively, we can also enhance the approximation by shrinking the size of $C_{\bm{a}}$. 
For example, suppose we want to obtain a more accurate approximation than the $0$-sublevel set of $h_3(a,b)$ over the domain $C_{\bm{a}}'=[-1,0]\times [-1,0]$.
We first partition the range $C_{\bm{a}}'$ into four smaller boxes.
Then, we solve Prog.~(\ref{eq:ex1-prog}) with respect to each  small box, still with relaxation order $d=3$.
For instance, over the box $[-0.5,0]\times [-0.5,0]$, we obtain
\begin{equation*}
    h'_3(a,b)=-0.5311 + \dots + 0.0478ab^2 + 0.54927b^3, 
\end{equation*}
where the range of $a,b$ is scaled to $[-1,1]\times[-1,1]$.
The $0$-sublevel set $h'_3(a,b)\le 0$ is depicted in \cref{fig:partition1}.
After scaling back, we can see the combination of under-approximations over these four smaller boxes cover the original approximation at $d=3$, as shown in \cref{fig:partition2}.
}

\begin{figure}[t]
\captionsetup{font={small}}
    \centering
    \begin{subfigure}[b]{0.22\textwidth}
        \centering
        \includegraphics[width=\textwidth]{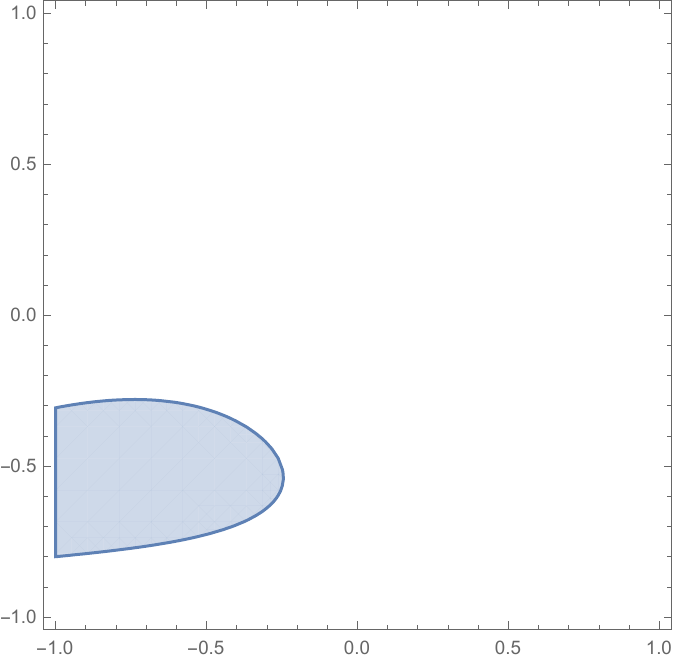}
        \caption{$d=3$}
    \end{subfigure}
    \begin{subfigure}[b]{0.22\textwidth}
        \centering
        \includegraphics[width=\textwidth]{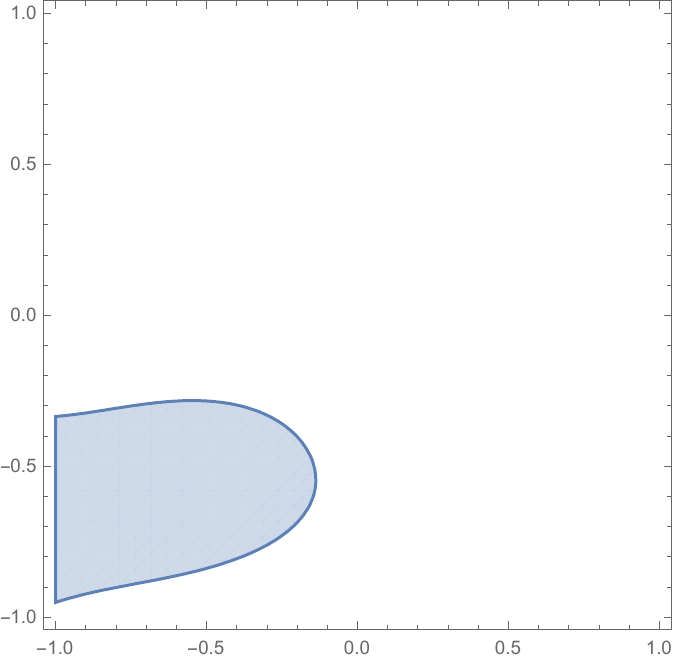}
        \caption{$d=4$}
    \end{subfigure}
    \begin{subfigure}[b]{0.22\textwidth}
        \centering
        \includegraphics[width=\textwidth]{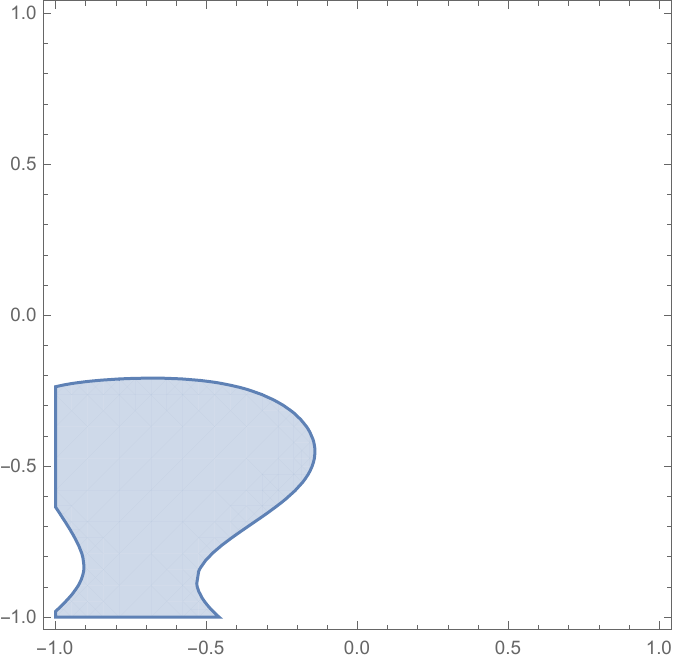}
        \caption{$d=5$}
    \end{subfigure}
    \begin{subfigure}[b]{0.22\textwidth}
        \centering
        \includegraphics[width=\textwidth]{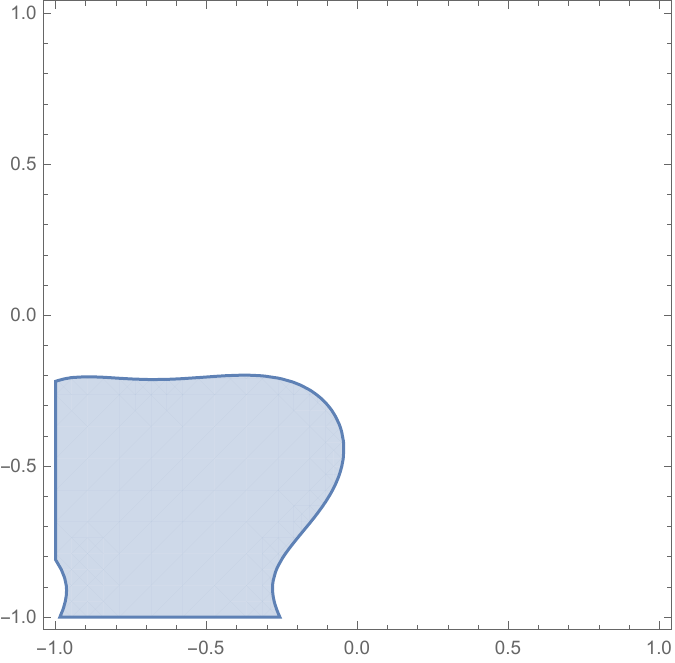}
        \caption{$d=6$}
    \end{subfigure}
    \\
    \scriptsize{
    x-axis: $a$ from -1.0 to 1.0; y-axis: $b$ from -1.0 to 1.0
    }
    \captionsetup{justification=centering}
    \caption{Under-approximations of {$R_I$} defined by 
     $R_{I,d}=\{(a,b)\in [-1,1]^2\mid h_{d}(a,b)\le 0\}$ for $d=3,4,5,6$.}
    \label{fig}
\end{figure}

\begin{figure}[t]
    \centering
        \begin{subfigure}[b]{0.22\textwidth}
        \centering
        \includegraphics[width=\textwidth]{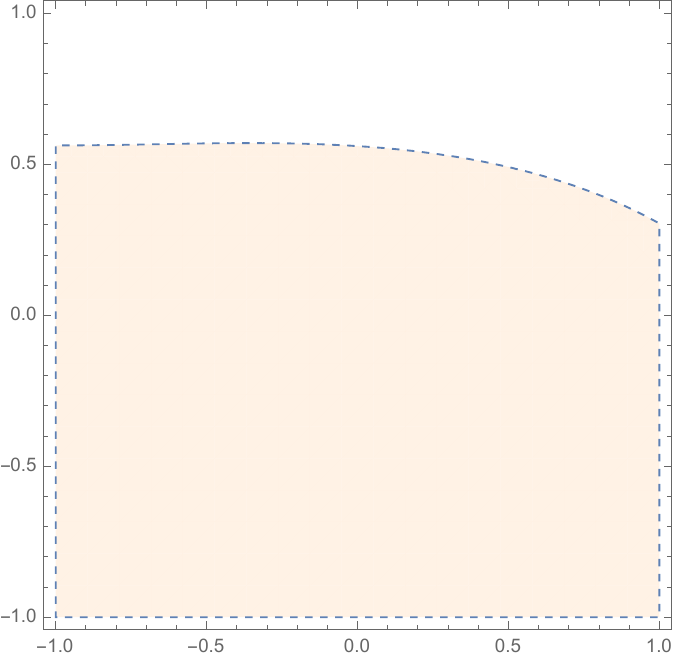}
        \caption{}
        \label{fig:partition1}
    \end{subfigure}
    \qquad 
    \begin{subfigure}[b]{0.22\textwidth}
        \centering
        \includegraphics[width=\textwidth]{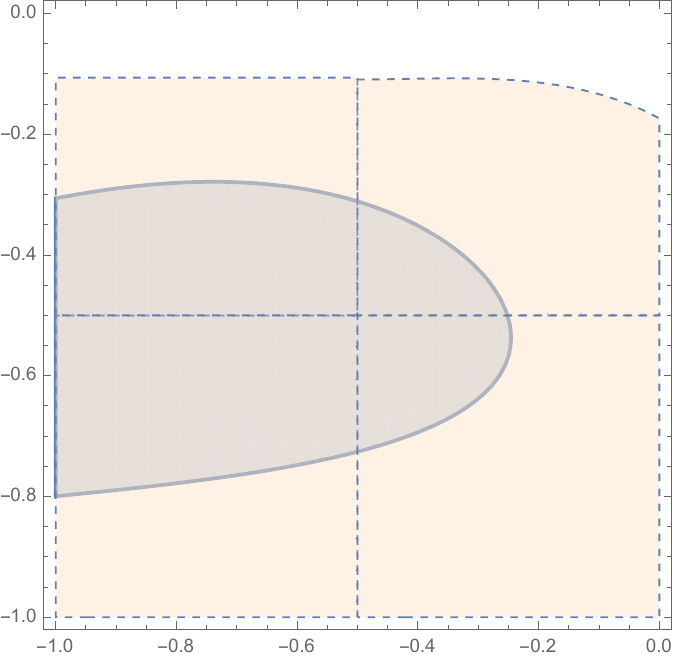}
        \caption{ }
        \label{fig:partition2} 
    \end{subfigure}
    \\
    \scriptsize{
    (a): $h'_3(a,b)\le 0$ over $[-0.5,0]\times [-0.5,0]$.\\
    (b): the comparison between $h_3(a,b)\le 0$ (light blue) and $h'_3(a,b)\le 0$ (light orange) over $[-1,0]\times [-1,0]$.
    }
    \captionsetup{justification=centering}
    \caption{A more accurate under-approximation of $R_I$ by partitioning.}
\end{figure}

Furthermore, using the solutions above, the weak invariant synthesis problem becomes significantly simpler. 
It can be reduced to solving a single polynomial inequality of the form $h_d(a,b)\le 0$ for $a,b\in C_{\bm{a}}$.
When the degree and the number of variables are not large, this type of problem can be efficiently addressed using general-purpose numerical optimization solvers or symbolic computation tools.
For example, when $d=3$, a solution $a=-1$, $b=-0.79971$ is returned by \textsf{FindInstance} function in \tool{Mathematica},
which implies
\begin{equation*}
    x^2 - 10 y^2 - 7.9971 \le 0
\end{equation*}
is an invariant of Code.~\ref{code:illustritive}.
Even though this example seems simple enough, many existing invariant synthesizing tools do not support this nonlinear template or fail to synthesize a suitable invariant~\cite{joel2018polynomial, CSS03CAV, kincaid2017nonlinear}.
On the other hand, directly applying symbolic solvers such as \textsc{Z3} or \textsf{Redlog}~\cite{REDLOG} to constraints Eqs.~(\ref{eq:inv-pre})-(\ref{eq:inv-post}) also fails to produce a satisfying assignment of parameters in an hour. 

\end{example}

\section{Extensions of Cluster Algorithm}
\label{sec:extension}

In this section, 
we extend our approach to allow for more complex program structures and invariant templates.

\subsection{Nested Loops}
\label{sec:extension-2}
In \cref{sec:under-approximation},
we have focused on unnested conditional loops of the form Code~\ref{code:model}.
In fact, our approach can be extended to nested loops or even control flow graphs without substantial changes.

\begin{listing}[h]
\begin{minted}[mathescape, escapeinside=||]{c}
    // Program variables: $\bm{x}\in \mathbb{R}^n$
    // Precondition: $\mathit{Pre} = \{ \bm{x}\mid \bm{q}_{pre}(\bm{x}) \le 0$} 
    // Invariant (outer): $I_1(\bm{a},\bm{x})$ with $\bm{a}\in \mathbb R^{n'_1}$
    while (|$g_1(\bm{x})\le 0$|) {
        |$\bm{x} \gets \bm{f}_1(\bm{x})$|;
        // Invariant (inner): $I_2(\bm{b},\bm{x})$ with $\bm{b}\in \mathbb R^{n'_2}$
        while (|$g_2(\bm{x})\le 0$|) {
            |$\bm{x} \gets \bm{f}_2(\bm{x})$|;
        }        
    }
    // Postcondition: $\mathit{Post} = \{ \bm{x} \mid \bm{q}_{post}(\bm{x})\le 0\}$
\end{minted}
\caption{A Simple Nested Loop}
\label{code:nested}
\end{listing}

To illustrate the main idea, 
let us consider a simple nested loop of the form Code~{\ref{code:nested}}. 
Synthesizing invariants for Code~{\ref{code:nested}} is more challenging compared to Code~\ref{code:model} because it involves two invariants:
$I_1(\bm{a},\bm{x})\in \mathbb R[\bm{a},\bm{x}]$ for the outer while loop and $I_2(\bm{b},\bm{x})\in \mathbb R[\bm{b},\bm{x}]$ for the inner while loop.
Similar to the unnested cases, we assume that $\bm{a}\in C_{\bm{a}}, \bm{b}\in C_{\bm{b}}$ for some known hyper-rectangle $C_{\bm{a}}, C_{\bm{b}}$.
The goal is to find valid assignments of the parameters $\bm{a}$ and $\bm{b}$ satisfying the following constraints:
\begin{align}
    &\forall \bm{x} \in C_{\bm{x}}.~\bm{q}_{pre}(\bm{x})\le 0\implies I_1(\bm{a},\bm{x})\le 0,\label{eq:outer-pre}\\
    &\forall \bm{x} \in C_{\bm{x}}.~I_1(\bm{a},\bm{x})\le 0 \wedge g_1(\bm{x})\le 0 \implies I_2(\bm{b},\bm{f}_1(\bm{x}))\le 0, \label{eq:inner-pre}\\
    &\forall \bm{x} \in C_{\bm{x}}.~I_2(\bm{b},\bm{x})\le 0 \wedge g_2(\bm{x})\le 0 \implies I_2(\bm{b},\bm{f}_2(\bm{x}))\le 0, \label{eq:inner-ind}\\
    &\forall \bm{x} \in C_{\bm{x}}.~I_2(\bm{b},\bm{x})\le 0 \wedge  g_2(\bm{x})\ge 0 \implies I_1(\bm{a},\bm{x})\le 0,\label{eq:inner-post}\\    
    &\forall \bm{x} \in C_{\bm{x}}.~I_1(\bm{a},\bm{x})\le 0 \wedge  g_1(\bm{x})\ge 0 \implies \bm{q}_{post}(\bm{x})\le 0.\label{eq:outer-post}
\end{align}
Here Eqs.~(\ref{eq:inner-pre})-(\ref{eq:inner-post}) play two roles:
for the inner loop, they encode the conditions of $I_2$ to be an invariant with $I_1$ serving as both precondition and postcondition;
for the outer loop, they collectively encode the inductive condition of $I_1$.

We can observe that constraints Eqs.~(\ref{eq:outer-pre})-(\ref{eq:outer-post}) still adhere to the form of the constraint in Prog.~(\ref{eq:para-pop}), albeit parameters $\bm{a}$ are replaced by $(\bm{a},\bm{b})$.
As a result, our approach in \cref{sec:under-approximation} remains applicable, and the soundness, convergence, and semi-completeness results carry over.
Our approach can generally be applied to programs represented by control flow graphs.

\subsection{Semialgebraic Templates}
\label{sec:extension-3}
In this subsection, we discuss the extensions of our approach to deal with more general polynomial templates, called (basic) semialgebraic templates.

\revise{
\begin{definition}[(Basic) Semialgebraic Template]
\label{def:semi}
    A \emph{semialgebraic template} is a finite collection of polynomials 
    $I_{t,r}(\bm{a}, \bm{x}) \in \mathbb{R}[\bm{a}, \bm{x}]$ defined over $C_{\bm{a}} \times C_{\bm{x}}$, 
    where $C_{\bm{a}}\subseteq \mathbb{R}^{n'}$ is a hyper-rectangle.
    Given a parameter assignment $\bm{a}_0 \in \mathbb{R}^{n'}$, 
    the instantiation of the invariant $\mathit{Inv}$ w.r.t. $\bm{a}_0$ is
    the set ${\{ \bm{x} \mid \bigvee_{t} \bigwedge_{r} I_{t,r}(\bm{a}_0, \bm{x}) \leq 0 \}}$.
    If $t$ belongs to a singleton set, then the template is called a \emph{basic semialgebraic template}, and  the instantiation of the invariant is the set ${\{ \bm{x} \mid \bigwedge_{r} I_{r}(\bm{a}_0, \bm{x}) \leq 0 \}}$.
\end{definition}
}

The robustness of (basic) semialgebraic templates is defined similar to that of polynomial templates.
\begin{definition}[Robustness]\label{def:robust-semi}
A semialgebraic template $\{I_{t,r}(\bm{a},\bm{x})\}_{t,r}$ is said to be \emph{robust} (w.r.t. the program model Code.~\ref{code:model}) if there exists a valid parameter assignment $\bm{a}_0\in C_{\bm{a}}$ and a small constant $\epsilon>0$ such that any $\bm{a}$ satisfying $\|\bm{a}-\bm{a}_0\|<\epsilon$ is still valid.
\end{definition}

In the following, we briefly show that techniques in \cref{sec:under-approximation} can be directly applied to the cases when templates are basic semialgebraic (instead of only polynomial) without substantial changes. 
After that, we discuss how to deal with semialgebraic templates in the general form. 

When given a basic semialgebraic template or a semialgebraic template, the invariant conditions Eqs.~(\ref{eq:inv-pre})-(\ref{eq:inv-post}) should be modified by replacing $I(\bm{a}, \bm{x})$ with $\bigwedge_r I_r(\bm{a},\bm{x}) \le 0$ or $\bigvee_{t} \bigwedge_{r} I_{t,r}(\bm{a}_0, \bm{x}) \leq 0$, respectively.
Recall that the original constraints Eqs.~(\ref{eq:inv-pre})-(\ref{eq:inv-post}) take the form of Prog.~(\ref{eq:para-pop}), the new constraints can be viewed as replacing the constraint in Prog.~(\ref{eq:para-pop}) by 
\begin{equation} \label{eq:l-basic}
    \forall \bm{x}. \bigwedge_{i=1}^m p_i(\bm{a},\bm{x})\ge 0\implies \bigwedge_r l_r(\bm{a},\bm{x})\le 0,
\end{equation}
or 
\begin{equation} \label{eq:l-semi}
    \forall \bm{x}. \bigwedge_{i=1}^m p_i(\bm{a},\bm{x})\ge 0\implies \bigvee_{t} \bigwedge_{r} l_{t,r}(\bm{a}, \bm{x}) \leq 0,
\end{equation}
where $p_i(\bm{a},\bm{x}), l_r(\bm{a},\bm{x}), l_{t,r}(\bm{a},\bm{x})\in \Real[\bm{a},\bm{x}]$.

As for basic semialgebraic templates, the problem can be reduced to the polynomial case with minor modifications.
This is because the constraint \cref{eq:l-basic} can be rewritten as
\begin{equation}
    \forall \bm{x}. \bigwedge_{i=1}^m p_i(\bm{a},\bm{x})\ge 0\implies l_r(\bm{a},\bm{x})\le 0,\quad \text{ for each } r.
\end{equation}
Consequently, the derived SOS relaxations will be like Prog.~(\ref{eq:qsos-relax}) but will involve more constraints. 
After that, all other results can be derived similarly. 

For general semialgebraic templates, unfortunately, simply rewriting the constraints no longer works due to the existence of disjunction. 
To address this issue, we resort to the \emph{lifting} technique introduced in~\cite{lasserre2012beyond} to reduce \cref{eq:l-semi} to the form of Prog.~(\ref{eq:para-pop}).

First, note that \cref{eq:l-semi} is equivalent to 
\begin{equation}
    \forall \bm{x}. \bigwedge_{i=1}^m p_i(\bm{a},\bm{x})\ge 0\implies s(\bm{a},\bm{x}) \le 0,
\end{equation}
where $s(\bm{a},\bm{x})=\min_t \max_r l_{t,r}(\bm{a},\bm{x})$.

Let us say a function $f(\bm{x})$ is a semialgebraic function if its graph $\{ (\bm{x}, f(\bm{x})) \in \Real^{n+1}\mid \bm{x}\in \Real^n\}$ is a semialgebraic set.
By Tarski-Seidenberg principle~\cite[Prop.~2.2.4]{bochnak1998real} and the definition of $s(\bm{a},\bm{x})$, we can prove $s(\bm{a},\bm{x})$ is a semialgebraic function.
Since the graph of every semi-algebraic function is the projection of a basic semialgebraic set in the lifted space~\cite[Lem.~3]{lasserre10jo},
we know that the graph
\begin{equation}
    \set[\bigg]{\big(\bm{a},\bm{x},s(\bm{a},\bm{x})\big)\given \bigwedge_{i=1}^m p_i(\bm{a},\bm{x})\ge 0}
    \subset \mathbb R^{n+n'+1}
\end{equation}
is the projection of some basic semialgebraic set 
$\hat{\mathcal K}\subseteq \mathbb R^{n+n'+1+u}$ for some $u\in\mathbb N$, i.e., 
\begin{equation}
    \set[\bigg]{(\bm{a},\bm{x},v)\given v=s(\bm{a},\bm{x}) \wedge \bigwedge_{i=1}^m p_i(\bm{a},\bm{x})\ge 0 } = 
    \set[\bigg]{(\bm{a},\bm{x},v)\given \exists \bm{w}\in \mathbb R^{u}. (\bm{a},\bm{x},v,\bm{w})\in \hat{\mathcal{K}}}.
\end{equation}
where $v\in \Real$ and $\bm{w}\in \Real^u$ are fresh variables

Here, $\hat{\mathcal K}$ is called the \emph{lifting} of $\mathcal K$ and can be computed following the techniques in~\cite{lasserre2012beyond}.
For now, we assume that $\hat{\mathcal K}$ is already obtained as 
\begin{equation}
    \hat{\mathcal K} = \set[\bigg]{ (\bm{a},\bm{x},v,\bm{w}) 
    \given
    \bigwedge_{i=1}^{\hat{m}} \hat{p}_i(\bm{a},\bm{x},v,\bm{w})\ge 0},
\end{equation}
where $\hat{p}_i(\bm{a},\bm{x},v,\bm{w}) \in \Real[\bm{a},\bm{x},v,\bm{w}]$.
Then, \cref{eq:l-semi} is equivalent to the following constraint in higher dimensions:
\begin{equation}
    \forall (\bm{x},v,\bm{w}).~ 
    \bigwedge_{i=1}^m p_i(\bm{a},\bm{x})\ge 0 \wedge  \bigwedge_{i=1}^{\hat{m}} \hat{p}_i(\bm{a},\bm{x},v,\bm{w})\ge 0 \implies v \le 0,
\end{equation}
which conforms to the form of Prog.~(\ref{eq:para-pop}) when treating $v,\bm{w}$ equivalently as $\bm{x}$.

\begin{remark}
In practice, our algorithm is less efficient for general semialgebraic templates compared to polynomial and basic semialgebraic templates. 
The main reason lies in the lifting process, 
which dramatically increases either the degree of defining polynomials or the number of parameters, sometimes even both.     
\end{remark}

\section{Synthesizing Weak Invariants From Masked Templates}
\label{sec:masked}

Recall that the semi-completeness of our Cluster algorithm hinges on the assumption that the given invariant templates are robust.
In \cref{sec:5-1}, we delve into situations where this assumption is violated, which implies that the Cluster algorithm may fail to produce a non-trivial under-approximation of the valid set. 
In these scenarios, in \cref{sec:5-2}, we identify a special subclass of basic semialgebraic templates, called \emph{Masked templates}, consisting of parametric polynomial equalities and some known polynomial inequalities.
Regarding these templates,  \cref{sec:5-3} proposes the \textbf{Mask algorithm} to translate the invariant conditions by exploiting the structure of masked templates. The resulting constraints can also be solved by SOS relaxations.

\subsection{On Robustness of Semialgebraic Templates} \label{sec:5-1}

The semi-completeness result (\cref{thm:weakComplete}) of the Cluster algorithm relies on the assumption that the given invariant templates are robust. 
\revise{
When the valid set $R_I$ is non-empty but the assumption is violated, solving Prog.~(\ref{eq:qsos-inv-relax}) will never yield a non-empty under-approximation $R_{I,d}$ for any $d\in \Nat$.
In this case, the Cluster algorithm is ineffective.
}

\revise{
While checking the robustness of a template is decidable (see \cref{prop:decide}), the decision procedure involves quantifier elimination, a computationally expensive process.
Fortunately, we have the following empirical observation:
In most cases, when a basic semialgebraic template does not contain equalities, either the template is robust or the valid set $R_I$ is empty.
In this context, having an equality in a basic semialgebraic template means that, for some polynomial $I_i(\bm{a},\bm{x})\in \Real[\bm{a},\bm{x}]$, both $I_i(\bm{a},\bm{x})$ and $-I_i(\bm{a},\bm{x})$ are contained in the template. 
}

\revise{
To illustrate the intuition behind this observation, consider a simple loop of the form \begin{equation*}
\textsf{while~($-1\le 0$)~do~$\{\bm{x}\gets \bm{x}\}$}.
\end{equation*}
}Suppose that we are given a precondition $Pre = \{\bm{x} \mid \bm{x}-\bm{x_0}\le 0, \bm{x}_0-\bm{x}\le 0\}$ for some $\bm{x}_0\in \Real^n$ and a basic semialgebraic template 
\begin{equation}
    \{I_1(\bm{a},\bm{x}),-I_1(\bm{a},\bm{x}),\dots,I_m(\bm{a},\bm{x}),-I_m(\bm{a},\bm{x}),I_{m+1}(\bm{a},\bm{x}),\dots,I_{s}(\bm{a},\bm{x})\}
\end{equation}
where $\bm{a}$ is linear in $I_{i}(\bm{a},\bm{x})$ for $i=1,\dots,s$.
Since the loop never terminates, the postcondition is irrelevant.
In this setting, the invariant conditions Eqs.~(\ref{eq:inv-pre})-(\ref{eq:inv-post}) can be reduced to the following single constraint:
\begin{equation}
    \begin{aligned}
    \forall \bm{x}.~\bm{x}-\bm{x}_0= 0 \implies &I_1(\bm{a},\bm{x})=0\wedge \cdots \wedge I_m(\bm{a},\bm{x})=0 \\
    &\qquad \wedge I_{m+1}(\bm{a},\bm{x}) \le 0 \wedge \cdots \wedge I_{s}(\bm{a},\bm{x})\le 0, 
    \end{aligned}
\end{equation}
which is equivalent to a linear system of equations in parameters $\bm{a}$:
\begin{equation}
    \begin{aligned}
    I_i(\bm{a},\bm{x}_0) &= 0 \text{ for } i=1,\dots,m\\
    I_j(\bm{a},\bm{x}_0) &\le 0 \text{ for } j=m+1,\dots,s.   
    \end{aligned}
\end{equation}
Then, the valid set $R_I$ is exactly the solution set to this linear system.
Suppose the linear expressions $I_i(\bm{a},\bm{x}_0)$ are linearly independent.
Based on the standard knowledge of linear algebra, the dimension of the set of solutions is generally $n'-m$, where $n'$ is the dimension of $\bm{a}$.
\revise{
However, $R_I$ having an interior point means that $R_I$ is of dimension $n'$, which requires $m=0$, i.e., there are no equalities in the template.
}


In practice, when the precondition and the postcondition contain equalities, we tend to need a basic invariant template with equalities.
Even though the Cluster algorithm may fail in such situations, we can utilize the structures of the equalities to design more efficient algorithms.
This  motivates  the definition of masked templates and the Mask algorithm below.


\subsection{Masked Templates}\label{sec:5-2}

Before presenting the definitions, we fix some notations.
Recall that $\bm{f}_i:\mathbb R^n \to \mathbb R^n$ is the assignment function of the $i$th branch.
We slightly abuse the notation $\bm{f}_{i, \bm{y}}$ to denote the projection of $\bm{f}_i$ onto variables $\bm{y}\subseteq \bm{x}$.
For example, if $\bm{f}_{1}(x_1,x_2,x_3)=(x_1,x_2^2,x_3^3)$ and $\bm{y}=(x_2,x_3)$, then $\bm{f}_{1,\bm{y}}=(x_2^2,x_3^3)$.

\begin{definition}[Core Variables]
\label{def:core}
Given a loop as in Code~\ref{code:model}, if the program variables can be divided into two non-empty parts $\bm{x}=(\bm{y},\bm{z})$, where $\bm{y} \in \Real^{n_{\bm{y}}}$ and $\bm{z} \in \Real^{n_{\bm{z}}}$ with $n_{\bm{y}}+n_{\bm{z}}=n$, such that: 
\begin{enumerate}
    \item for each $y$ in $\bm{y}$, $\bm{f}_{i,y} \in \mathbb R[\bm{y}]$;
    \item for each $z$ in $\bm{z}$, $\bm{f}_{i,z}\in \mathbb R[\bm{y},\bm{z}]$ and is linear in $\bm{z}$;
    \item the loop guard $\bm{g}$ and branch conditionals $\bm{c}_i$  are independent of $\bm{z}$, and the postcondition $\bm{q}_{post}(\bm{y},\bm{z})$ is linear in $\bm{z}$.
\end{enumerate}
then we call $\bm{y}$ \emph{core variables} and $\bm{z}$ \emph{non-core variables}.
\end{definition}
\begin{definition}[Masked Template]
\label{def:masked}
    Given core variables $\bm{y}$ and non-core variables $\bm{z}$,
    a masked template is a finite collection of polynomials $I_r(\bm{a}, \bm{y}) \in \mathbb{R}[\bm{a}, \bm{y}]$ linear in $\bm{a}$ and polynomials $I_t(\bm{y})\in \mathbb{R}[\bm{y}]$, 
    \revise{where $r\in \{1,\dots,n_{\bm{z}}\}$ and $t\in T$ for some finite index set $T$}. 
    Given a parameter assignment $\bm{a}_0 \in \mathbb{R}^{n'}$, 
    the instantiation of the invariant $\mathit{Inv}$ w.r.t. $\bm{a}_0$ is
    the set $\{\bm{x} \mid \bigwedge_{r=1}^{n_{\bm{z}}} z_r = I_r(\bm{a}_0, \bm{y}) \wedge \bigwedge_{t\in T} I_t(\bm{y})\le 0 \}$.
\end{definition}

In other words, the program variables are partitioned into two non-empty parts, core variables~$\bm{y}$ and non-core variables~$\bm{z}$;
and the non-core variables $\bm{z}$ will only occur in the precondition, the postcondition, and their assignment functions.
In a masked template, we wish to express non-core variables $\bm{z}$ by core variables $\bm{y}$.

\begin{example}\label{ex:2}
Code~\ref{code:freire1} is the algorithm for finding the closest integer (variable $r$) to the square root (of variable $y$), taken from the benchmark set~\cite{rodriguez2016some}. 
Here, we treat integers as a subset of real numbers.

\begin{listing}[h]
\begin{minted}[mathescape, escapeinside=||]{c}
    // Program variables: $(x,y,r)\in \Real^n$
    // Precondition: $Pre = \{(x,y,r) \mid -y\le 0, x=y/2, r=0\}$
    // Invariant template: {$y= {\color{purple} poly[(x,r),2]}, -x\le 0$}
    while (|$r - x \le 0$|) { // Real invariant: $y= 2x+r^2-r \wedge x\ge 0$
        |$x \gets x - r$|;
        |$y \gets y$|;
        |$r \gets r + 1$|;
    }
    // Postcondition: $Post = \{ y -  r^2 - r \le 0, r^2 - r - y\le 0\}$
\end{minted}
\caption{\textsf{freire1}}
\label{code:freire1}
\end{listing}

In this program, we can view $(x,r)$ as core variables and $y$ as the only non-core variable.  
In the invariant template, the purple part represents some unknown polynomial in $x,r$ of degree $2$. 
In other words, we are given a masked template of the form
$\{y = I(\bm{a},x,r), - x\le 0\}$, where $I(\bm{a},x,r)=a_1 x^2 + a_2 x r+ a_3 r^2 + a_4 x + a_5 r + a_6$ with parameters $\bm a=(a_1,\dots,a_6)$.

In this case, the valid set $R_I$ does not contain an interior point because the following constraint
\begin{equation*}
    \forall (x,y,r)\in C_{\bm{x}}.~
    y \ge 0 \wedge x=y/2 \wedge r=0 \implies y = I(\bm{a},x,r) \wedge x\ge 0
\end{equation*}
implies that $a_1=a_6=0$ and $a_4=2$.
\end{example}

Our observation suggests that core variables often correspond to local variables (such as $x,r$), while non-core variables tend to align with the input arguments (such as $y$).
The motivation for the definition of masked templates is that, for most programs, the invariants include two parts: 
(i) equalities of the form that a part of variables are expressed by the other variables;
(ii) inequalities that are derived from conditionals and monotonicity.
Besides, the name of masked templates is inspired by the so-called ``masked programs'' in~\cite{GHM+23OOPSLA}.

\subsection{Mask Algorithm} \label{sec:5-3}

In this part, we show how to transform the invariant conditions for masked templates based on variable substitution.
The outstanding property of the resulting constraints is that they can be directly solved by SOS relaxations.

Given a masked template as in \cref{def:masked}, we explicitly write down the invariant conditions w.r.t. Eqs.~(\ref{eq:inv-pre})-(\ref{eq:inv-post}), with $\bm{x}$ replaced by $(\bm{y},\bm{z})$ 
\begin{align}
    \forall (\bm{y},\bm{z}) \in C_{\bm{x}}.~
    &\bm{q}_{pre}(\bm{y},\bm{z})\le \bm{0} 
    \implies \bigwedge_{r=1}^{n_{\bm{z}}} z_r=I_r(\bm{a},\bm{y}) \wedge \bigwedge_{t\in T} I_t(\bm{y})\le 0\label{eq:mask-inv-pre}\\
    \forall (\bm{y},\bm{z}) \in C_{\bm{x}}.~
    &\bigwedge_{r=1}^{n_{\bm{z}}} z_r = I_r(\bm{a},\bm{y}) \wedge \bigwedge_{t\in T} I_t(\bm{y})\le 0 \wedge \bm{g}(\bm{y})\le \bm{0} \wedge \bm{c}_i(\bm{y})\le \bm{0} \notag\\
    &\qquad \implies 
    \bigwedge_{r=1}^{n_{\bm{z}}} f_{i,z_r}(\bm{y},\bm{z}) = I_r(\bm{a},\bm{f}_{i,\bm{y}}(\bm{y})) \wedge \bigwedge_{t\in T} I_t(\bm{f}_{i,\bm{y}}(\bm{y}))\le 0, \quad i=1,\dots,k,\label{eq:mask-inv-induc}\\
    \forall (\bm{y},\bm{z}) \in C_{\bm{x}}.~
    &\bigwedge_{r=1}^{n_{\bm{z}}} z_r = I_r(\bm{a},\bm{y}) \wedge \bigwedge_{t\in T} I_t(\bm{y})\le 0 \wedge \neg( \bm{g}(\bm{y})\le 0) \notag\\
    &\qquad \implies \bm{q}_{post}(\bm{y},\bm{z})\le 0.\label{eq:mask-inv-post}
\end{align}
where we omit $\bm{z}$ in $\bm{g}(\bm{y},\bm{z})$, $c_i(\bm{y},\bm{z})$, and $\bm{f}_{i,\bm{y}}(\bm{y},\bm{z})$.

As discussed in Remark.~\ref{remark:bmi}, 
the SOS relaxations of \cref{eq:mask-inv-induc} and \cref{eq:mask-inv-post} can not be translated into SDPs because parameters $\bm{a}$ occur in the left-hand-side of the implications.
However, after a simple variable substitution procedure, the above constraints can be converted into a desired form.
The substitution is based on the following observation in first-order logic:
the formula
\begin{equation}
    \forall y,z.~\big (z=f(y)\wedge A(y)\big ) \implies B(y,z),\label{eq:fol1}
\end{equation}
is equivalent to 
\begin{equation}
    \forall y.~A(y)\implies B(y,f(y)),\label{eq:fol2}
\end{equation}
where $f$ is a function and $A,B$ are formulas. 

Exploiting this idea, \cref{eq:mask-inv-induc} and \cref{eq:mask-inv-post} can be transformed into
\begin{align}
    \forall \bm{y} \in C_{\bm{y}}.~
    &\bigwedge_{t\in T} I_t(\bm{y})\le 0 \wedge \bm{g}(\bm{y})\le \bm{0} \wedge \bm{c}_i(\bm{y})\le \bm{0} \notag\\
    &\qquad \implies 
    \bigwedge_{r=1}^{n_{\bm{z'}}} f_{i,z_r}(\bm{y},\bm{z'}) = I_r\big(\bm{a},\bm{f}_{i,\bm{y}}(\bm{y})\big) \wedge \bigwedge_{t\in T} I_t\big(\bm{f}_i(\bm{y})\big)\le 0, \quad i=1,\dots,k,\label{eq:mask-inv-induc-sub}\\
    \forall \bm{y}\in C_{\bm{y}}.~
    &\bigwedge_{t\in T} I_t(\bm{y})\le 0 \wedge \neg( \bm{g}(\bm{y})\le 0) \implies \bm{q}_{post}(\bm{y},\bm{z'})\le 0.\label{eq:mask-inv-post-sub}
\end{align}
where $\bm{z'} = (I_1(\bm{a},\bm{y}), \dots, I_{n_{\bm{z}}}(\bm{a},\bm{y}))$ and $C_y$ is the domain of $\bm{y}$.

\revise{
Therefore, finding a valid parameter assignment for  Eqs.~(\ref{eq:mask-inv-pre}), (\ref{eq:mask-inv-induc}) and (\ref{eq:mask-inv-post}) is reduced to solving the following program:
\begin{equation}\label{eq:mask}
    \begin{aligned}
        \text{find} \quad &\bm{a}\\
        s.t. \quad & \text{Constraints Eqs.~(\ref{eq:mask-inv-pre}), (\ref{eq:mask-inv-induc-sub}), and (\ref{eq:mask-inv-post-sub})}.
    \end{aligned}
\end{equation}
According to the definition of masked templates, constraints \cref{eq:mask-inv-induc-sub} and \cref{eq:mask-inv-post-sub} contain no nonlinear terms in parameters $\bm{a}$. 
Therefore, Prog.~(\ref{eq:mask}) conforms to the form of  Prog.~(\ref{eq:pop}) and can be solved by using the standard SOS relaxation techniques in \cref{sec:pre-sos}.
As a consequence, the soundness result (\cref{thm:pre-sound}) and semi-completeness result (\cref{thm:pre-complete}) carry over.
}

\IncMargin{1em}
\begin{algorithm2e}[t]
\SetKwInOut{Input}{Input}
\SetKwInOut{Output}{Output}
\ResetInOut{Output} 
\Input{A program $\mathcal{P}$ of the form Code~\ref{code:model}, a masked template, and an upper bound $D\in \mathbb{N}$ on the relaxation order.}
\Output{A valid parameter assignment $\bm{a}_0$.}
\BlankLine
Construct Prog.~(\ref{eq:mask}) using $\mathcal{P}$ and the masked template\;
$d_{\max} \gets$ the largest degree of polynomials in Prog.~(\ref{eq:mask})\;
$d_r \gets \lfloor \frac{d_{\max}+1}{2}\rfloor$\;
\While{
    $d_r \le D$
}{
\eIf{
the $d_r$-th SOS relaxation of Prog.~(\ref{eq:mask}) is solvable
}{
$\bm{a}_0 \gets$ Solve the $d_r$-th SOS relaxation of Prog.~(\ref{eq:mask})\;
\KwRet{$\bm{a}_0$}\algocomment*[l]{a valid parameter assignment}
}{
    $d_r\gets d_r+1$\;
}
}
\KwRet{$\emptyset$}\algocomment*[l]{either $D$ is not large enough or the template has no solution}
\caption{\revise{The Mask Algorithm}}
\label{alg:mask}
\end{algorithm2e}
\DecMargin{1em}

\begin{manualexample}{2}[Continued]
After performing a variable substitution, we construct Prog.~(\ref{eq:mask}) as follows:
\begin{align*}
    \text{find} \quad & \bm{a}\\
    s.t. \quad &\forall (x,y,r)\in C_{x,y,r}.~
    y \ge 0 \wedge x=y/2 \wedge r=0 \implies y = I(\bm{a},x,r) \wedge x\ge 0,\\
    &\forall (x,r)\in C_{x,r}.~
    x\ge 0 \wedge x \ge r \implies I(\bm{a}, x,r)=I(\bm{a},x-r,r+1) \wedge x\ge 0,\\
    &\forall (x,r)\in C_{x,r}.~
    x\ge 0 \wedge x \le r \implies r^2+r\ge I(\bm{a},x, r) \wedge r^2-r\le I(\bm{a},x, r).    
\end{align*}

In particular, by solving the $2$nd SOS relaxation of the above constraints,
we obtain the following numerical result:
\revise{
\begin{align*}
    I(\bm{a}_0,x,r) = & ~ 0.0000000020+1.9999999314\cdot x-0.9999999624\cdot r+\\ 
    & ~ 0.9999999665\cdot r^2-0.0000000041\cdot x^2-0.0000000098\cdot x\cdot r,
\end{align*}
where we round off the coefficients to 10 decimal places just to demonstrate the numerical errors.
In our experiments, we use a technique called \emph{rationalization} (explained later) to eliminate the numerical errors, obtaining:}
\begin{equation*}
I(\bm{a}_0,x,r)=2x+r^2-r,    
\end{equation*}
which can be verified to be an invariant.
\end{manualexample}

\section{Experiments}
\label{sec:experimentsnew}
\revise{
\paragraph{Research Questions}
Our experiments aim to answer the following research questions:
\begin{itemize}
    \item \textbf{RQ1}: How does the \textbf{Cluster algorithm} perform on \emph{strong} invariant synthesis problems with a relatively \emph{small} number of parameters?
    \item \textbf{RQ2}: How does the \textbf{Mask algorithm} perform on \emph{weak} invariant synthesis problems with a relatively \emph{large} number of parameters?
\end{itemize}

Note that the number of parameters significantly impacts the computational complexity of our two algorithms.
For Cluster algorithm, Prog.~(\ref{eq:qsos-inv-relax}) involves SOS polynomials with Gram matrices of size $\binom{n+n'+d_r}{d_r}$.
To make it tractable, we need to keep the sum $n+n'$ (i.e., the number of $\bm{x}$ and $\bm{a}$) small. 
In contrast, Prog.~(\ref{eq:mask}) of the Mask algorithm involves SOS polynomials with Gram matrices of size $\binom{n_{\bm{y}}+d_r}{d_r}$, suggesting that it could be more scalable for problems with a larger number of parameters.
}


\paragraph{Implementation}
We have developed prototypical implementations of our two SDP-based algorithms in \textsc{Matlab} (R2020a),
interfaced with \textsc{Yalmip}~\cite{Lofberg2004} and \textsc{Mosek}~\cite{DBLP:journals/mp/AndersenRT03} to solve the underlying SOS relaxations.
The implementation and benchmarks can be found at  {\color{blue} \url{https://github.com/EcstasyH/invSDP}}. 
All experiments were performed on a 2.50GHz Intel Core i9-12900H laptop running 64-bit Windows 11 with 16GB of RAM and Nvidia GeForce RTX 3060 GPU.

\paragraph{Comparison}
\revise{
Since the strong invariant synthesis algorithms in~\cite{kapur05inv} and~\cite{chatterjee2020polynomial} are not implemented, we mainly compare with the weak invariant synthesis tools.
For our Cluster algorithm, we add one extra step to synthesize weak invariants: When an under-approximation $R_{I,d}=\{\bm{a}\in C_{\bm{a}} \mid h_d(\bm{a})\le 0\}$ is synthesized, we use the symbolic tool \tool{Mathematica} to solve $h_d(\bm{a})\le 0$ (see the comment on line 13 in \cref{alg:cluster}).
}

We wish to compare our tool with the most relevant tool in~\cite{chatterjee2020polynomial}, but, unfortunately, their implementation is not publicly available.
Instead, we primarily compare with \tool{PolySynth}~\cite{GHM+23OOPSLA}, which is a recent template-based synthesis tool that supports both generations of programs and invariants.
When focusing on synthesizing (weak) invariants, \tool{PolySynth} adopts the same strategy as in~\cite{chatterjee2020polynomial} to encode the invariant conditions into quadratic constraints.
It also employs additional techniques and heuristics for further speedup.
To compare with~\cite{kapur05inv}, which is based on quantifier elimination, we try solving the constraints Eqs.~(\ref{eq:inv-pre})-(\ref{eq:inv-post}) in Z3~\cite{Z3}.
Since most problem instances used in our experiments have basic semialgebraic templates, we do not compare with~\cite{LWY+14FCS,AGM15SAS}, which only support polynomial templates.
Apart from template-based approaches, we also provide a comparison with the state-of-the-art machine learning approach LIPuS~\cite{YWW23ISSTA}.

\paragraph{Benchmarks}
Our two algorithms are applicable to different situations depending on whether the basic semialgebraic templates contain equalities.
Regarding this, we design two sets of benchmarks:

\begin{itemize}
    \item \textbf{Cluster benchmarks} include two groups of problem instances, modified programs and dynamical systems.
    (1) The modified programs are obtained from the corresponding programs in the masked benchmarks by relaxing the specifications and adjusting the templates. 
    Such modification ensures that the selected templates are robust so that the Cluster algorithm should be able to find a solution, when $D$ is large enough.
    \revise{
    We also include two benchmarks, \textsf{freire1-2} and \textsf{cohencu-2}, with templates quadratic in parameters $\bm{a}$, to validate the statement in Remark~\ref{remark:parameter}. 
    The other benchmarks use templates that are linear in parameters~$\bm{a}$ by default.
    }
    (2) For programs abstracted from dynamical systems in literature, we try to search for ellipsoid-shaped invariants and do not know whether the template is robust.
    In all these examples, we restrict the number of parameters $\bm{a}$ to not exceed 3. 
    Since \tool{PolySynth} lacks support for specified domains, we do not supply $C_{\bm{x}}$ to it in these benchmarks. 
    For these benchmarks, we primarily compare our tool with complete approaches \tool{PolySynth} and Z3, as LIPuS does not support floating-point data. 
    \item \textbf{Masked benchmarks} \revise{include programs without nested loops from~\cite{rodriguez2016some}.
    Some benchmarks (like \textsf{euclidex2} and \textsf{cohendiv}) are adapted so that they do not involve integer operations (such as $\gcd(x,y)$).    
    We also added a manually constructed parameterized benchmark \textsf{sum-k-power-d}, which can be found at the end of this section.
    }
    In these benchmarks, the program variables and constants are mostly integers, treated as a subset of $\Real$. 
    For these programs, as in \cref{ex:2}, we set the invariant templates to contain all monomials of core variables up to a given degree (determined  by the real invariants obtained by experts). 
    \revise{
    Moreover, since program variables are unbounded in these benchmarks, we do not add the bound restriction (i.e., $C_{\bm{x}}=\Real^n$) and the Mask algorithm remains sound (see \cref{remark:homo}).
    }
    For these benchmarks, the comparisons encompass \tool{PolySynth}, LIPuS and Z3.
\end{itemize}

For our two algorithms and \tool{Z3}, we manually encode the information of the programs and the templates.
For \tool{PolySynth}, the inputs are the programs and the invariant templates.  
\tool{LIPuS} does not need templates and the inputs are just the programs.
The experimental results are reported in \cref{tab:1} and \cref{tab:2}, respectively.

\begin{remark}
It is worth noting that the results of \tool{PolySynth} and \tool{LIPuS} are not directly comparable with the results reported in their original papers, as the specifications and templates can be different.
For example, in \textsf{freire1} example, the template in our experiment is given as in \cref{ex:2} with $6$ parameters, while in~\cite{GHM+23OOPSLA} they assume $2$ parameters in the invariant, i.e., $y=a_1x+a_2+r^2-r$.
For many other examples in~\cite{GHM+23OOPSLA}, such as \textsf{berkeley}, they assume the invariants are known and try to synthesize unknown parameters in the loop body.
\end{remark}

\begin{table}[t]
\captionsetup{font={small}}
\caption{Experimental results over Cluster benchmarks.}
\label{tab:1}
\begin{center}
    \begin{tabular}{l cc ccr c cr c cr} 
    \toprule
    ~ & ~ & ~ &\multicolumn{3}{c}{Cluster~(\cref{alg:cluster})}& ~ 
      &\multicolumn{2}{c}{\tool{PolySynth}~\cite{GHM+23OOPSLA}}& ~
      &\multicolumn{2}{c}{Z3~\cite{Z3}}\\ 
    \cmidrule{4-6} \cmidrule{8-9} \cmidrule{11-12}
    Benchmark & $n_{\bm{x}}$ & $n_{\bm{a}}$ & $D$ &result & time & ~ & result & time & ~ & result & time\\ 
    \midrule
    \textsf{freire1-1} & 3 & 2 & 1 & \greencheck & 0.7s & ~ & \greencheck & 3.6s & ~ & \greencheck & \textbf{0.1s}\\
    \textsf{freire1-2} & 3 & 2& 2 & \greencheck & \textbf{4.2s} & ~ & \myred{TO} & >600s & ~ & \myred{TO} & >600s\\
    \textsf{freire1-3} & 3 & 2& 1 & \greencheck & \textbf{0.8s} & ~ & \greencheck & 3.1s & ~ & \greencheck & 61.0s\\
    \textsf{cohencu-1} & 2 & 3 & 2 & \greencheck & 5.4s & ~ & \myred{TO} & >600s & ~ & \greencheck & \textbf{0.2s}\\
    \textsf{cohencu-2} & 2 & 3 & 3 & \greencheck & \textbf{9.4s} & ~ & \myred{TO} & >600s & ~ & \myred{TO} & >600s\\
    \textsf{cohencu-3} & 2 & 3 & 3 & \greencheck & 11.2s & ~ & \myred{TO} & >600s & ~ & \greencheck & \textbf{0.2s}\\
    \midrule
    \cref{ex:1} & 2 & 2 & 3 &\greencheck & \textbf{3.0s} & ~ & \myred{TO} & >600s & ~ & \myred{TO} & >600s \\
    \textsf{circuit}~\cite{AMT+21} & 2 & 2 & 2 &\greencheck & \textbf{2.8s} & ~ & \myred{TO} & >600s & ~ & \myred{TO} & >600s \\
    \textsf{unicycle}~\cite{sassi2012controller} & 2 & 3 & $\ge 7$ & \myred{TO} & >600s & ~ & \myred{TO} & >600s & ~ & \myred{TO} & >600s \\
    \textsf{overview}~\cite{dai2013nonlinear} & 2 & 3 & $\ge 8$ &\myred{TO} & >600s & ~ & \myred{TO} & >600s & ~ & \myred{TO} & >600s \\
    \bottomrule
    \end{tabular}
    \end{center}
    \scriptsize{
    $n_{\bm{x}}, n_{\bm{a}}$: the number of variables and parameters, respectively.   
    For the Cluster algorithm: $D$ is the smallest degree upper bound such that a non-empty under-approximation can be found,
    ``time'' includes the solving SDP time and the posterior verification time. 
    \myred{TO}: timeout, 600s.
    The boldface marks the winner.
    }
\end{table}

\revise{
\paragraph{Numeric vs Symbolic} 
One important feature of our approach is that we use a numerical solver. 
While benefiting from the efficiency of numerical algorithms, the produced results may be unsound due to numerical errors. For example, as shown in \cref{ex:2}, there is a tiny gap between the numerical result and the real invariant. 
\revise{
Though there exist exact solver~\cite{Henrion18,henrion19-spectra} and arbitrary-precision solver~\cite{JoldesMP17}, our experience shows that these tools do not scale for the SDP problems translated from our benchmarks.
}
In our experiments, we employ the following strategies to deal with numerical errors:
\begin{itemize}
    \item \textbf{Rounding Off}: For the Cluster algorithm, since the invariant template is robust, the obtained valid parameter assignment $\bm{a}_0$ is usually an interior point of the valid set.
    This makes the result tolerant to small numerical errors. Even with slight perturbations, $\bm{a}_0+\epsilon$ remains valid for small $\epsilon$.
    While SDP solvers usually offer accuracies around $10^{-8}$~\cite{Roux16},
    there is no guarantee of the accuracy of solutions for SOS relaxations. 
    In our experience, rounding off to five decimal places can achieve relatively accurate results.
    \item \textbf{Rationalization}: 
    For the Mask algorithm, the invariant template is not robust as equalities are involved. 
    As a result, solely relying on the rounding off strategy may produce incorrect answers.
    In addition, the benchmarks contain problem instances that require rational coefficients like $\frac{5}{12}$, which can not be expressed by floating point numbers.
    To this end, we employ the idea from~\cite{kaltofen12jsc}: we first rationalize the numerical result, then verify its correctness. 
    The rationalization is achieved using the built-in function \textsf{rat} in \textsf{Matlab}, which returns the rational fraction approximation of the input to within a specified tolerance ($10^{-5}$).
    For example, when a numerical result $0.41667$ is obtained, the function \textsf{rat} transforms it into a continued fractional expansion $0 + 1/(2 + 1/(3 + 1/(-2)))$, which can then be simplified to $\frac{5}{12}$.
    \item \textbf{Posterior Verification}: 
    The above two strategies can not address all numerical problems. 
    To guarantee soundness, when a numerical solution is obtained, we use \tool{Mathematica} to verify the correctness of the corresponding invariant candidate. 
\end{itemize}
}

\begin{table}[t]
\captionsetup{font={small}}
\caption{Experimental results over masked template benchmarks.}
\label{tab:2}
\begin{center}
    \begin{tabular}{c cccc cc c cc c cc} 
    \toprule
    ~ & ~ & ~ & ~ & ~ &\multicolumn{2}{c}{Mask~(\cref{alg:mask})}& ~ 
  &\multicolumn{2}{c}{\tool{PolySynth}~\cite{GHM+23OOPSLA}} & ~
  &\multicolumn{2}{c}{LIPuS~\cite{YWW23ISSTA}}  \\ 
    \cmidrule{6-7} \cmidrule{9-10} \cmidrule{12-13} 
    Benchmark & $n_{\bm{y}}$ & $n_{\bm{z}}$ & $n_{\bm{a}}$ & $k$ & result & time & ~ & result & time & ~ & result & time$\dagger$ \\
    \midrule
    \textsf{berkeley} & 4 & 1 & 5& 4& \greencheck & 3.8s & ~ & \myred{TO} & >600s & ~ & \myred{TO} & >600s\\
    \textsf{cohencu}  & 1 & 3 & 12& 1 & \greencheck & 1.7s & ~ & \myred{TO} & >600s
    & ~ & \myred{TO} & >600s\\
    \textsf{cohendiv}  & 4 & 1 & 15& 1 & \greencheck & 0.6s & ~ & \greencheck & 6.8s & ~ & (\greencheck) & (142.0s)\\
    \textsf{euclidex2}  & 6 & 2 & 56& 2 &\greencheck & 4.0s & ~ & \myred{TO} & >600s& ~ & \myred{TO} & >600s\\
    \textsf{fermat2}  & 4 & 1 & 15& 2 &\greencheck & 0.9s & ~ & \greencheck & 10.3s& ~ & \myred{TO} & >600s\\
    \textsf{firefly}  & 4 & 1 & 5& 7 & \greencheck & 10.9s & ~ & \myred{TO} & >600s& ~ & \myred{TO} & >600s\\
    \textsf{freire1}  & 2 & 1 & 6& 1 &\greencheck & 0.7s & ~ & \greencheck & 70.4s& ~ & (\greencheck) & (460.0s)\\
    \textsf{freire2}  & 4 & - & -& 1 & \myred{NS} & - & ~ & \myred{TO} & >600s  & ~ & \myred{NS} & - \\
    \textsf{illinois}  & 4 & 1 & 5& 10 &\greencheck & 17.0s & ~ & \myred{TO} & >600s& ~ & \myred{TO} & >600s\\
    \textsf{lcm}  & 6 & 1 & 28 & 2 & \greencheck & 1.3s & ~ & \greencheck & 17.1s& ~ & \myred{TO} & >600s\\
    \textsf{mannadiv}  & 4 & 1 & 15& 3 & \greencheck & 1.2s & ~ & \myred{TO} & >600s& ~ & \myred{TO} & >600s\\
    \textsf{mesi}  & 4 & 1 & 5& 4 &\greencheck & 4.5s & ~ & \myred{TO} & >600s& ~ & (\greencheck) & (592.5s)\\
    \textsf{moesi}  & 5 & 1 & 6& 5 &\greencheck & 7.7s & ~ & \myred{TO} & >600s& ~ & (\greencheck) & (117.2s)\\
    \textsf{petter}  & 1 & 1 & 7& 1 &\greencheck & 0.6s & ~ & \greencheck & 4.4s& ~ & \myred{TO} & >600s\\
    \textsf{readerswriters}  & 5 & 1 & 21& 4 & \greencheck & 3.4s & ~ & \myred{TO} & >600s& ~ & \myred{TO} & >600s\\
    \textsf{sqrt}  & 4 & - & -& 1 &\myred{NS} &  - & ~ & \myred{TO} & >600s& ~ & (\greencheck) & (421.0s)\\
    \textsf{wensley}  & 7 & - & -& 2 & \myred{NS} &  - & ~ & \myred{TO} & >600s & ~ & \myred{NS} &  - \\
    \textsf{z3sqrt}  & 4 & 1 & 15& 2& \greencheck & 1.1s & ~ & \myred{TO} & >600s  & ~ & \myred{NS} &  - \\ \hline
    %
    \textsf{sum2power10}  & 2 & 1 & 66& 1 &\greencheck & 8.6s & ~ & \myred{TO} & >600s  & ~ & \myred{TO} &  >600s \\
    \textsf{sum2power15}  & 2 & 1 & 136& 1 &\myred{\xmark} & 200.3s & ~ & \myred{TO} & >600s  & ~ & \myred{TO} &  >600s \\
    \textsf{sum3power6}  & 3 & 1 & 84 & 1 &\greencheck & 9.9s & ~ & \myred{TO} & >600s  & ~ & \myred{TO} &  >600s \\
    \textsf{sum3power8}  & 3 & 1 & 102 & 1 &\myred{\xmark} &  102.4 & ~ & \myred{TO} & >600s  & ~ & \myred{TO} &  >600s \\
    \textsf{sum5power4}  & 5 & 1 & 126 & 1 &\greencheck &  26.6s & ~ & \myred{TO} & >600s  & ~ & \myred{TO} &  >600s \\
    \textsf{sum5power5}  & 5 & 1 & 252 & 1 &\myred{TO} &  >600s & ~ & \myred{TO} & >600s  & ~ & \myred{TO} &  >600s \\
    \textsf{sum8power3}  & 8 & 1 & 165 & 1 &\greencheck & 184.6s & ~ & \myred{TO} & >600s  & ~ & \myred{TO} &  >600s \\
    \bottomrule
    \end{tabular}
    \end{center}
    \scriptsize{
    $n_{\bm{y}}, n_{\bm{z}}, n_{\bm{a}}$: the number of core variables, non-core variables, and parameters in templates, respectively. $k$: the number of branches in loop body.
    For the Mask algorithm, ``time'' includes the solving SDP time and the posterior verification time. 
    \myred{TO}: timeout, 600s. \myred{NS}: unsupported benchmarks where there are no non-core variables (Mask algorithm) or containing floating-point variables (not allowed in LIPuS). 
    \myred{\xmark}: fail to synthesize an invariant or unsound invariant.
    $\dagger$: the results in the last column were provided by the authors of LIPuS using their computational environment in~\cite{YWW23ISSTA}.
    }
\end{table}

\paragraph{Experimental Results over Cluster Benchmarks} 
For all benchmarks in the first group, our algorithm successfully synthesized a non-empty under-approximation $\{\bm{a}\mid h_D(\bm{a})\le 0\}$ with $D\le 3$, but there were a few cases where directly using Z3 outperformed our approach.
Note that both \tool{PolySynth} and Z3 failed in synthesizing a valid invariant for \textsf{freire1-2} and \textsf{cohencu-2}, this was possibly attributed to the fact that the templates are quadratic in parameters $\bm{a}$.
Although the problems in the first group are relatively easy, they already pose a challenge for \tool{PolySynth} and Z3. 
As for the second group, our algorithm succeeded in the first two problem instances, while the other two tools failed on all problems.

\revise{
This answers \textbf{RQ1}: The Cluster algorithm can efficiently synthesize under-approximations to strong invariant synthesis problems with a relatively small number of parameters. Additionally, its outputs simplify the weak invariant synthesis problem.
However, our approach also has limitations. 
}The Cluster algorithm lacks a termination criterion when a nonempty under-approximation cannot be found. 
This means that the algorithm might run indefinitely in such cases.
In addition, as the degree bound $D$  increases, the time required to solve both the SDP and the polynomial inequality $h_D(\bm{a})\le 0$ also grows significantly.

\paragraph{Experimental Results over Masked Benchmarks}
In the first group of benchmarks, our algorithm demonstrated its effectiveness by successfully synthesizing valid invariants for 14 out of 18 problem instances. 
The runtimes of our algorithm for this set of benchmarks were typically under 10 seconds, demonstrating its practical applicability and efficiency.
However, the \textsf{cohencu} benchmark failed due to a small numerical error (approximately $10^{-5}$).  
In the three unsupported benchmarks, we were unable to identify the core variables according to our definition.
In comparison, \tool{PolySynth} solved only five instances, while LIPuS and Z3 (omitted from the table) failed to produce results for all instances in 10 minutes.
To fully understand LIPuS's capacity, we collaborated with its creators to evaluate our benchmarks within their environment. 
The results are displayed in the final column of Table~\ref{tab:2}, distinguished by parentheses.

The second group of benchmarks demonstrates the scalability of the Mask algorithm, as it can handle problem instances with up to a few hundred parameters within a 10-minute timeout.
However, its performance degrades when the degree of the template increases due to numerical issues in the underlying solvers. This is evident in cases like \textsf{sum2power15} and \textsf{sum3power8}, where the solver incorrectly returns "infeasible" despite the existence of a real invariant. 
The \textsf{sumpower} benchmarks are tailored to the Mask algorithm. Therefore, a direct comparison with other tools might not be entirely fair.

\revise{
Overall, the experimental results answer \textbf{RQ2}: Compared to existing methods, the Mask algorithm demonstrated superior efficiency and scalability on benchmarks with a relatively large number of parameters.
}

\paragraph{Sum-k-Power-d Benchmark}
This is a manually constructed example with no practical meanings, solely intended to demonstrate how the number of parameters will influence the efficiency of our algorithm.
It is easy to see $n_{\bm{a}}=\binom{k+d}{d}$.

\begin{listing}[ht]
\begin{minted}[mathescape, escapeinside=||]{c}
    // Program variables: $(n_1,n_2,\dots,n_k,s)\in \Real^n$
    // Precondition: $s=(n_1+n_2+\dots+n_k)^d$
    // Invariant template: {$s= {\color{purple} poly[(n_1,\dots,n_k),d]}$}
    while ( true ) { // Real invariant $s=(n_1+n_2+\dots+n_k)^d$
        |$n_1 \gets n_1 + 1$|;
        |$\dots$|;
        |$n_k \gets n_k + 1$|;
        |$s \gets (n_1+n_2+\dots+n_k+k)^d - (n_1+n_2+\dots+n_k)^d$|;
    }
    // Postcondition: true
\end{minted}
\caption{\textsf{sum-$k$-power-$d$}}
\label{code:sumpower}
\end{listing}


\section{Related Work}
\label{sec:related}

In this section, we present different methods for invariant synthesis and compare our approaches with the most related works.

\paragraph{Constraint Solving}
As constraint-solving techniques have made significant advancements in recent years, constraint-solving-based approaches have become increasingly relevant and promising.
Specifically, for synthesizing linear invariants, \cite{CSS03CAV} proposes the first complete approach based on Farkas' lemma, which can be seen as a linear version of Putinar's Positivstellensatz (\cref{thm:putinar}), and solves the resulting nonlinear constraints by quantifier elimination.
However, due to the doubly exponential time complexity of quantifier elimination procedures~\cite{DH88}, this method is impractical even for programs of moderate size.
Therefore, many works consider using heuristics to solve the nonlinear constraints for better scalability~\cite{SSM04SAS,LFY+22OOPSLA}.

The problem of synthesizing polynomial invariants for polynomial programs is more challenging.
For both the weak and the strong invariant synthesis problem, \cite{kapur05inv} introduces the first complete approach based on quantifier elimination.
For the strong invariant synthesis problem, when coefficients in templates are rational numbers, \cite{chatterjee2020polynomial} shows that the complexity bound can be improved to subexponential in the length of the template. 
For the weak invariant synthesis problem, 
subsequent works can be broadly categorized into two classes: 
One group focuses on efficiently solving the general constraints of invariant conditions~\cite{YZZ+10FCSC,chatterjee2020polynomial,GHM+23OOPSLA}, while the other group strengthens the invariant conditions to make the constraints easier to solve~\cite{Cousot05VMCAI, LWY+14FCS,AGM15SAS}. 
\revise{A concise overview of these approaches is presented in Table~\ref{tab:sum}.}  
In the following, we compare the technical differences between our work and some related works which also use~\cref{thm:putinar}. 

Comparison with~\cite{chatterjee2020polynomial,GHM+23OOPSLA}: 
These two papers provide a systematic way to encode the invariant conditions of programs represented by control-flow graphs.
As discussed in Remark~\ref{remark:bmi}, they also employ~\cref{thm:putinar} to translate the invariant conditions into constraints involving SOS polynomials (essentially bilinear matrix inequalities). 
To handle these constraints, they further encode them into quadratic programs and rely on general-purpose solvers.
\revise{
In contrast, our algorithms encode constraints into SDPs.
This is achieved by utilizing Lasserre's technique~\cite{lasserre15} (in the Cluster algorithm) and by exploiting specific patterns in templates (in the Mask algorithm).
}

Comparison with~\cite{LWY+14FCS,AGM15SAS}: 
These two papers deal with the weak invariant synthesis problem by strengthening the invariant conditions in the form of Prog.~(\ref{eq:sos}), allowing for standard SOS relaxations.
However, their techniques are limited to polynomial templates, which means the invariant must be the 0-sublevel set of a single polynomial and do not have completeness guarantees.
As a result, these approaches can not synthesize invariants for programs in our masked template benchmarks. 


\begin{table}[t]
\captionsetup{font={small}}
\caption{Summary of constraint-solving-based approaches for polynomial invariant synthesis.}
\label{tab:sum}
\centering
    \begin{tabular}{| c |c | c | c | c | c | c | c |} 
    \hline
    \textbf{Invariant} &\textbf{Algorithm} & \textbf{Loop} & \textbf{Template} & \textbf{Constraint} & \textbf{Guarantee} & \textbf{Size of $\sigma$}\\ \hline 
    \hline
    \multirow{2}{*}{Strong} & Cluster (\cref{sec:under-approximation}) & nested & basic semi.  & SDP  & convergence$^*$ & $\binom{n+n'+d_r}{d_r}$\\\cline{2-7} 
    & \cite{kapur05inv,chatterjee2020polynomial} & nested & basic semi. & FOL & complete & - \\\hline\hline 
    \multirow{3}{*}{Weak} & Mask (\cref{sec:masked}) & simple & masked & SDP  & semi-complete  & $\binom{n_{\bm{y}}+d_r}{d_r}$\\\cline{2-7}
    & \cite{LWY+14FCS,AGM15SAS} & simple & poly. & SDP & - & $\binom{n+d_r}{d_r}$ \\\cline{2-7}
    & \cite{chatterjee2020polynomial, GHM+23OOPSLA} & nested & basic semi. & QP & semi-complete & $ \binom{n+d_r}{d_r}$ \\\hline
 \end{tabular}
 \\ \scriptsize{
    \textbf{Invariant}: the strong or weak invariant synthesis problem.
    \textbf{Loop}: the structure of loop models: ``nested'' means nested loops, and ``simple'' means non-nested loops. 
    \textbf{Template}: the type of invariant templates.
    \textbf{Constraint}: the form of encoded constraints: ``QP'' means quadratic programming,  ``FOL'' means first-order logic in reals.  
    \textbf{Guarantee}: the theoretical guarantee (all these approaches are sound), and the asterisk ($^*$) means the robustness assumption is needed.
    \textbf{Size of $\sigma$}: the size of Gram matrices of SOS polynomials (measuring the magnitude of the constraints for methods based on \cref{thm:putinar}), where $d_r$ is the relaxation order and  $n$, $n'$, and $n_{\bm{y}}$ are the dimensions of $\bm{x}$, $\bm{a}$, and $\bm{y}$, respectively.
    }
\end{table}

\paragraph{Craig Interpolation}
Craig interpolation is a power tool for local and modular reasoning.
In first-order logic, if a formula $P$ implies a formula $Q$, then there exists a formula $I$, called a \emph{Craig interpolation} or simply \emph{interpolation}, such that $P$ implies $I$, $I$ implies $Q$, and every non-logical symbol in $I$ occurs in both $P$ and $Q$.
\revise{
In program verification, interpolants can serve as invariants, though they may not always be inductive.
\cite{lin17ase-squeeze} introduces an algorithmic framework for generating interpolants and subsequently strengthening them into inductive invariants. 
Similar approaches have been explored in~\cite{dai2013nonlinear,ijcar16,Gan2020} for synthesizing nonlinear invariants for polynomial programs. 
Craig interpolation has also been employed in model checking techniques for invariant generation, leading to a diverse range of algorithms~\cite{mcmillan06cav,mcmillan03cav,komuravelli14cav-recursive,cimatti16fmsd,hojjat18fmcad,blicha22tacas}.
}

\paragraph{Abstract Interpretation}
Abstract interpretation is a widely used and classic method for invariant generation~\cite{OS04IPL, rck04a, BRZ05SAS, rck07, AGG12LMCS}. 
The process involves fixing an abstract domain and iteratively performing forward propagation until a fixed point is reached, which serves as an invariant.
The effectiveness and efficiency of abstract interpretation approaches heavily rely on the choice of abstract domains. 
Different abstract domains may lead to varying levels of precision and scalability in the obtained invariants.
In most cases, there is no theoretical guarantee of the accuracy of generated invariants. 
In other words, it is uncertain whether the obtained invariant is strong enough to accurately represent the desired properties of the system under analysis. 
The absence of such guarantees necessitates careful consideration of the abstract domains and fine-tuning of the analysis to strike a balance between precision and tractability.

\revise{
\paragraph{Recurrence Analysis}
Recurrence-based methods \cite{rc04issac,rck07,Kovacs08TACAS,humenberger2018invariant,farzan15fmcad, kincaid2017nonlinear} typically involve these steps: 
(1) extracting recurrences from loops; 
(2) computing closed-form solutions for loop variables; and 
(3) deriving (equality) invariants from the solutions.
One major limitation of these methods is that they are restricted to loops with \emph{solvable mappings}~\cite{rc04issac,rck07} (and minor generalizations thereof), a restricted subclass of loops where the underlying recurrences are solvable.
Intuitively, solvable mappings generalize affine mappings by allowing certain acyclic nonlinear dependencies between variables.
Recent advancements~\cite{amrollahi22sas,cyphert24popl-ideal,wang-lin24cav} aim to handle the cases where recurrences for loop variables do not exist or are not solvable, by employing template-based ideas to generate recurrences for expressions in variables.
}

%


\paragraph{Other Methods}
Recently, methods based on machine learning~\cite{HSP+20PLDI, si2020code2inv, yao2020learning, YWW23ISSTA} and  logical inference~\cite{DDL+13OOPSLA,SA14CAV,PWK+22POPL,KPS+22TACAS} have shown significant promise.
Beyond classic programs, the problem of invariant generation is also being actively explored in the context of hybrid systems~\cite{DGX+17JSC,WCX+22IC,SP2023arXiv} and stochastic systems~\cite{CHW+15CAV, BTP+22CAV,BCJ+23TACAS}, combining techniques from differential equations and probability theory.





\revise{
We would like to emphasize that our definition of strong invariants differs from the concept of ``strongest (affine/polynomial) invariants'' commonly used in research on computability results, such as~\cite{karr76,OS04ICALP,joel2018polynomial,hrushovski23jacm,mullner24popl}.
These studies typically focus on computing sets of affine or polynomial \emph{equalities} that serve as loop invariants. In contrast, our work considers \emph{inequalities} of specific parameterized forms. For a comprehensive overview of these related results, we recommend consulting~\cite{mullner24popl}.
}

\section{Conclusions and Future Work}
\label{sec:conclusion}
\revise{
In this paper, we present two novel SDP-based approaches to synthesize invariants for polynomial programs, expanding the boundaries of constraint-solving-based invariant synthesis methods.
For the strong invariant synthesis problem, the Cluster algorithm employs a technique from robust optimization~\cite{lasserre15} to under-approximate the valid set.
For the weak invariant synthesis problem, the Mask algorithm relies on identifying special structures in program invariants.
Both our algorithms are sound and semi-complete.
}

Currently, the Cluster algorithm becomes impractical when the template includes an excessive number of parameters.
This limitation arises because the size of SOS polynomials in the relaxations depends on the total number of program variables and parameters.
To address this problem, we consider exploring the internal structure of the constraints to improve the algorithm.
Moreover, we plan to extend the techniques presented in this paper 
to invariant synthesis for hybrid systems and probabilistic programs. 




\begin{acks}
We are grateful to the anonymous reviewers for their insightful feedback and constructive criticism, which significantly improved the quality of this paper.
We also thank Shiwen Yu for his assistance in producing the benchmark results of LIPuS.
\end{acks}

\bibliographystyle{ACM-Reference-Format}
\bibliography{inv}



\end{document}